%% file: connectivity-talg.tex
\documentclass[11pt]{article}
\usepackage{verbatim,url,enumerate,color,paralist,graphicx}
\usepackage{amsmath,amsfonts,amssymb,amstext,amsthm}
\usepackage{fullpage}
\usepackage{hyperref,tabularx}

\newtheorem{theorem}{Theorem}[section]
\newtheorem{lemma}[theorem]{Lemma}

\newtheorem{definition}[theorem]{Definition}

\newtheorem{proposition}[theorem]{Proposition}

\newtheorem{myclaim}[theorem]{Claim}

\theoremstyle{definition}
\newtheorem{remark}[theorem]{Remark}

 \def\cqed{\renewcommand{\qedsymbol}{$\lrcorner$}}

 \newcommand{\ka}{\ensuremath{k}}
 \newcommand{\pe}{\ensuremath{p}}
 \newcommand{\dunion}{\dot{\cup}}

\begin{document}

\newif\ifappendix
\newif\ifmain
\newif\ifmove
\maintrue
\movefalse
\newif\ifabstract
\newif\iffull

 \abstractfalse

\ifabstract \fullfalse \else \fulltrue \movetrue \fi 

\ifabstract
\newcommand{\appstar}{$[*]$}
\else
\newcommand{\appstar}{}
\fi

\title{Fixed-parameter algorithms for minimum cost edge-connectivity augmentation}


\author{
D\'aniel Marx\thanks{
{Institute for Computer Science and Control,}
{Hungarian Academy of Sciences (MTA SZTAKI),}
{Budapest, Hungary,}
{dmarx@cs.bme.hu) Research
    supported by the European Research Council (ERC) grant
    ``PARAMTIGHT: Parameterized complexity and the search for tight
    complexity results,'' reference 280152.}
}
\and
L\'aszl\'o A.\ V\'egh\thanks{
       {Dept.\ of Management,}
       {London School of Economics \& Political Science,}
       {Houghton Street, London WC2A~2AE, UK.}
       {(l.vegh@lse.ac.uk)}
}}


\maketitle
\begin{abstract}
  We consider connectivity-augmentation problems in a setting where
  each potential new edge has a nonnegative cost associated with it, and
  the task is to achieve a certain connectivity  target with at most $p$ new
  edges of minimum total cost.  The main result is that the minimum
  cost augmentation of edge-connectivity from $k-1$ to $k$ with at
  most $p$ new edges is fixed-parameter tractable parameterized by $p$
  and admits a polynomial kernel. We also prove the fixed-parameter
  tractability of increasing edge-connectivity from 0 to 2, and
  increasing node-connectivity from 1 to 2.
\end{abstract}

\appendixtrue

\input{intro-v7}
\input{prelim-v7}
\input{byone-v7}

\input{arbitrary-v7}

\bibliographystyle{abbrv}
\bibliography{connectivity-fixed}



\end{document}

%% file: intro-v7.tex
\section{Introduction}

Designing networks satisfying certain connectivity requirements has
been a rich source of computational problems since the earliest days
of algorithmic graph theory: for example, the original motivation of
Bor\r uvka's work on finding minimum cost spanning trees was designing an
efficient electricity network in Moravia \cite{DBLP:journals/dm/NesetrilMN01}. In many applications, we
have stronger requirements than simply achieving connectivity: one may
want to have connections between (certain pairs of) nodes even after a
certain number of node or link failures. Survivable network design
problems deal with such more general requirements.

In the simplest scenario, the task is to achieve $k$-edge-connectivity
or $k$-node-connectivity by adding the minimum number of new edges
to a given directed or undirected graph $G$. This setting already
leads to a surprisingly complex theory and, somewhat unexpectedly,
there are exact polynomial-time algorithms for many of these
questions.  For example, there is a polynomial-time algorithm for
achieving $k$-edge-connectivity in an undirected graph by adding the
minimum number of edges (Watanabe and Nakamura \cite{watanabenakamura},
see also Frank \cite{frank-edge}). For
$k$-node-connectivity, a polynomial-time algorithm is known only for
the special case when the graph is already $(k-1)$-node-connected;
the general case is still open \cite{Vegh11}. We refer the reader to
the recent book by Frank \cite{frankkonyv} on more results of
similar flavour.
One can observe that increasing
connectivity by one already poses significant challenges and in
general the node-connectivity versions of these problems seem to be
more difficult than their edge-connectivity counterparts.

For most applications, minimizing the number of new edges is a very
simplified objective: for example, it might not be possible to realize
direct connections between nodes that are very far from each other. A
slightly more realistic setting is to assume that the input specifies
a list of potential new edges (``links'') and the task is to achieve
the required connectivity by using the minimum number of links from
this list. Unfortunately, almost all problems of this form turn out to
be NP-hard: deciding if the empty graph on $n$ nodes can be augmented
to be 2-edge-connected with $n$ new edges from a given list is
equivalent to finding a Hamiltonian cycle (similar simple arguments
can show the NP-hardness of augmenting to $k$-edge-connectivity also
for larger $k$). Even though these problems are already hard, this
setting is still unrealistic: it is difficult to imagine any
application where all the potential new links have the same
cost. Therefore, one typically tries to solve a minimum cost version
of the problem, where for every pair $u,v$ of nodes, a (finite or
infinite) cost $c(u,v)$ of connecting $u$ and $v$ is given. When the
goal is to achieve $k$-edge connectivity, we call this problem {\em
  Minimum Cost Edge-Connectivity Augmentation to $k$} (see
Section~\ref{sec:prelim} for a more formal definition). In the special
case when the input graph is assumed to be $(k-1)$-edge-connected (as
in, e.g., \cite{jordan95,hsu2000,kortsarz07,Vegh11}), we call the
problem {\em Minimum Cost Edge-Connectivity Augmentation by
  One}. Alternatively, one can think of this problem with the
edge-connectivity target being the minimum cut value of the input
graph plus one.  The same terminology will be used for the
node-connectivity versions and the minimum cardinality variants (where
every cost is either $1$ or infinite).

Due to the hardness of the more general minimum cost problems,
research over the last two decades has focused mostly on the approximability of the
problem. This field is also known as survivable network design,
e.g., \cite{agrawal95,goemans95,Jain01,CVV03,kortsarz03,cheriyan12};
for a survey, see \cite{kortsarz07}.
 In this paper, we approach these problems
from the viewpoint of parameterized complexity.
We say that a problem with parameter $p$ is
{\em fixed-parameter tractable (FPT)} if it can be solved in time
$f(p)\cdot n^{O(1)}$, where $f(p)$ is an arbitrary computable function
depending only on $p$ and $n$ is the size of the input
\cite{MR2001b:68042,MR2238686}. The tool box of fixed-parameter
tractability includes many techniques such as bounded search trees,
color coding, bidimensionality, etc. The method that received most
attention in recent years is the technique of kernelization
\cite{DBLP:conf/birthday/LokshtanovMS12,DBLP:journals/disopt/MisraRS11}. A
{\em polynomial kernelization} is a polynomial-time algorithm that
produces an equivalent instance of size $p^{O(1)}$, i.e., polynomial
in the parameter, but not depending on the size of the
instance. Clearly, polynomial kernelization implies fixed-parameter
tractability, as kernelization in time $n^{O(1)}$ followed by any
brute force algorithm on the $p^{O(1)}$-size kernel yields a
$f(p)\cdot n^{O(1)}$ time algorithm. The conceptual message of
polynomial kernelization is that the hard problem can be solved by
first applying a preprocessing to extract a ``hard core'' and then
solving this small hard instance by whatever method available.
An interesting example of fixed-parameter tractability in the context
of connectivity augmentation is the result
by Jackson and Jord\'an \cite{jacksonjordan}, showing that for the
problem of making a graph $k$-node-connected by adding a minimum
number of arbitrary new edges
admits a  $2^{O(k)}\cdot n^{O(1)}$ time algorithm (it is still open whether there is a polynomial-time algorithm for this problem).

As observed above, if the link between arbitrary pair of nodes is
not always available (or if they have different costs for different
pairs), then the problem for augmenting a $(k-1)$-edge-connected graph to a $k$-edge-connected one is NP-hard  for any fixed $k\ge 2$. Thus
for these problems we cannot expect fixed-parameter tractability when
parameterizing by $k$. In this paper, we consider a different
parameterization: we assume that the input contains an integer $p$,
which is a upper bound on the number of new edges that can be added.
Assuming that the number $p$ of new links is much smaller than the size of the
graph, exponential dependence on $p$ is still
acceptable, as long as the running time depends only polynomially on
the size of the graph. 
It follows from Nagamochi~\cite[Lemma
7]{DBLP:journals/dam/Nagamochi03} that {\em Minimum Cardinality
Edge-Connectivity Augmentation from 1 to 2} is fixed-parameter
tractable parameterized by this upper bound $p$. Guo and Uhlmann
\cite{DBLP:journals/networks/GuoU10} showed that this problem, as well
as its node-connectivity counterpart, admits a kernel of $O(p^2)$
nodes and $O(p^2)$ links. Neither of these algorithms seem to work
for the more general minimum cost version of the problem, as the
algorithms rely on discarding links that can be replaced by more
useful ones. Arguments of this form cannot be generalized to the case when
the links have different costs, as the more useful links can have
higher costs.  Our results go beyond the results of
\cite{DBLP:journals/dam/Nagamochi03,DBLP:journals/networks/GuoU10} by
considering higher order edge-connectivity and by allowing arbitrary
costs on the links.

We present a kernelization algorithm for the problem {\em Minimum Cost Edge-Connectivity Augmentation by One} for arbitrary $k$.
The algorithm starts by doing the opposite of the
obvious: instead of decreasing the size of the instance by discarding
provably unnecessary links, we add new links to ensure that the instance has a
certain closure property; we call instances satisfying this property  {\em metric instances}. We argue that these changes do not affect
the value of the optimum solution. Then we show that a metric instance has a bounded number of important links that are provably sufficient for the construction of an optimum solution.
The natural machinery for this approach via metric instances is to work with a more general
problem. Besides the costs, every link is equipped with a positive
integer weight. Our task is to find a minimum cost set of links of
total weight at most $p$ whose addition makes the graph
$k$-edge-connected. Our main result addresses the corresponding problem, {\em Weighted
  Minimum Cost Edge-Connectivity Augmentation}.
\begin{theorem}\label{th:main-weighted}
Weighted Minimum Cost Edge-Connectivity Augmentation by One admits a kernel of
$O(p)$ nodes, $O(p)$ edges, $O(p^3)$ links, with all costs being integers of $O(\pe^6\log
\pe)$ bits.
\end{theorem}

The original problem is the special case when all links have weight
one. Strictly speaking, Theorem~\ref{th:main-weighted} does not give a
kernel for the original problem, as the kernel may contain links of
higher weight even if all links in the input had weight one. Our next
theorem, which can be derived from the previous one,
shows that we may obtain a kernel that is an unweighted
instance. However, there is a trade-off in the bound on the kernel size.

\begin{theorem}\label{th:main}
Minimum Cost Edge-Connectivity Augmentation by One admits a kernel of
$O(p^4)$ nodes, $O(p^4)$ edges and $O(p^4)$ links, with all costs being  integers of $O(\pe^8\log
\pe)$ bits.
\end{theorem}

\medskip

Let us now outline the main ideas of the proof of Theorem~\ref{th:main-weighted}.
We first show that every input can be efficiently reduced to
a metric instance, one with the closure property.
We first describe our algorithm in the special case of increasing
edge-connec\-tivity from 1 to 2, where connectivity augmentation can be
interpreted as covering a tree by paths. 
The closure
property of the instance allows us to prove that there is an optimum
solution where every new link is incident only to ``corner nodes'' 
(leaves and branch nodes).  Either the
problem is infeasible, or we can
bound the number of corner nodes by $O(p)$. Hence we can also bound the
number of potential links in the resulting small instance.


Augmenting edge connectivity from 2 to 3 is similar to augmenting from
1 to 2, but this time the graph we need to work on is no longer a
tree, but a cactus graph. Thus the arguments are slightly more
complicated, but generally go along the same lines. Finally, in the
general case of increasing edge-connectivity from $k-1$ to $k$, we use
the uncrossing properties of minimum cuts and a classical result of
Dinits, Karzanov, and Lomonosov \cite{dinits} to show that (depending on the parity of $k$) the problem can be always
reduced to the case $k=2$
or $k=3$.

In kernels for the weighted problem, a further technical issue has to be
overcome: each finite cost in the produced instance has to be a
rational number represented by $p^{O(1)}$ bits. As we have no assumption
on the sizes of the numbers appearing in the input, this is a nontrivial requirement. It turns out that
a technique of Frank and Tardos~\cite{franktardos} (used
earlier in the design of strongly polynomial-time algorithms) can be
straightforwardly applied here: the costs in the input can be
preprocessed in a way that the each number is an integer of $O(p^6\log
p)$ bits long and the relative costs of the feasible solutions do not
change. We believe that this observation is of independent interest,
as this technique seems to be an essential tool for kernelization of
problems involving costs.

To prove Theorem~\ref{th:main} (see Section~\ref{sec:main-proof}),
we first obtain a kernel by applying our weighted result to our unweighted instance; this kernel will however contain links of weight higher than one.
Still, every link $f$ of weight $w(f)$ 
in the (weighted) kernel can be replaced by a sequence of $w(f)$ original unweighted edges. This replaces the $O(p^3)$ links
by $O(p^4)$ original ones. 

We try to extend our results in two directions. First, we show that in
the case of increasing connectivity from 1 to 2, the
node-connectivity version can be directly reduced to the edge-connectivity
version \iffull  (see Section \ref{sec:node}).\else (see Appendix \ref{sec:node}). \fi

\begin{theorem}\label{th:node}
Weighted Minimum Cost
Node-Connectivity Augmentation from $1$ to $2$ admits a
a kernel of
$O(p)$ nodes, $O(p)$ edges, $O(p^3)$ links, with all costs being integers of $O(\pe^6\log
\pe)$ bits.
\end{theorem}

For higher connectivities, we do not expect such a clean reduction to
work. Polynomial-time exact and approximation algorithms for
node-connectivity are typically much more involved than for
edge-connectivity (compare e.g., \cite{watanabenakamura} and
\cite{frank-edge} to \cite{frankjordan} and \cite{Vegh11}), and it is
reasonable to expect that the situation is similar in the case of
fixed-parameter tractability.

A natural goal for future work is trying to remove the assumption of
Theorems~\ref{th:main-weighted} and \ref{th:main} that the input graph is $(k-1)$-connected. In
the case of 2-edge-connectivity, we show that the problem is
fixed-parameter tractable even if the input graph is not
connected. However, the algorithm uses nontrivial branching and it
does not provide a polynomial kernel.

\begin{theorem}\label{th:0to2}
Minimum Cost
Edge-Connectivity Augmentation to $2$ can be solved in time $2^{O(p\log p)}\cdot n^{O(1)}$. 
\end{theorem}
The proof is given in Section~\ref{sec:augm-arbitr-appendix}. 
The additional branching arguments needed in Theorem~\ref{th:0to2} can
show a glimpse of the difficulties one can encounter when trying to
solve the problem larger $k$, especially with respect to
kernelization.  For augmentation by one, the following notion of
shadows was crucial to define the metric closure of the instances: $f$
is a shadow of link $e$ if the weight of $e$ is at most that of $f$,
and $e$ covers every $k$-cut covered by $f$ --- in other words,
substituting link $f$ by link $e$ retains the same connectivity. When
the input graph is not assumed to be connected, we cannot extend the
shadow relation to  links connecting different components, only in
special, restricted situations. Therefore, we cannot prove  the
existence of an optimal solution with all links incident to corner
nodes only. Instead, we prove that there is an optimal solution such
that all leaves are adjacent to either corner nodes or certain other
special nodes; this enables the branching in the FPT algorithm. A
further difficulty  arises if we want to avoid using two copies of the
same link. This was automatically excluded for augmentation by one,
whereas now further efforts are needed to enforce this requirement.




%% file: prelim-v7.tex
\ifmain
\section{Preliminaries}\label{sec:prelim}
For a set $V$, let ${V\choose 2}$ denote the edge set of the complete
graph on $V$. Let $n=|V|$ denote the number of nodes.
 For a set $X\subseteq V$ and $F\subseteq {V\choose 2}$, let $d_F(X)$
denote the number of edges $e=uv\in F$ with $u\in X$, $v\in V\setminus X$. When we are given a
graph $G=(V,E)$ and it is clear from the context, $d(X)$ will denote
$d_E(X)$. A set $\emptyset\neq X\subsetneq V$ will be called a {\em cut}, and {\em minimum cut} if
$d(X)$ takes the minimum value. 
For a function $z:V\rightarrow \mathbb{R}$, and a set $X\subseteq V$, let
$z(X)=\sum_{v\in X} z(v)$ (we use the same notation with functions
on edges as well). For $u,v\in V$, a set $X\subseteq V$ is called an
$u\bar{v}$-set if $u\in X$, $v\in V\setminus X$.

Let us be given an undirected graph $G=(V,E)$ (possibly containing
parallel edges), a connectivity target $\ka\in\mathbb{Z}_+$, and  a cost function $c:{V\choose 2}\rightarrow
\mathbb{R}_+\cup\{\infty\}$.
For a given nonnegative integer \pe{}, our aim is to find a minimum cost set of edges $F\subseteq {V\choose 2}$ of cardinality at most \pe{} 
such that $(V,E\cup F)$ is \ka-edge-connected. 

We will work with a more general version of this problem. 
Let $E^*$ denote an edge set on $V$, possibly containing parallel
edges. We call the elements of $E$ {\em edges} and the elements of $E^*$ {\em
  links}.  
Besides the cost function $c:E^*\rightarrow
\mathbb{R}_+\cup\{\infty\}$, we are also given a
positive integer weight function $w:E^*\rightarrow \mathbb{Z}_+$. 
We restrict the total weight of the augmenting edge set to be at most
$p$ instead of restricting its cardinality.
Let us define our main problem.

\medskip

\begin{center}
\fbox{
\parbox{0.85\linewidth}{
\smallskip

\noindent
 Weighted Minimum Cost Edge Connectivity
  Augmentation
\vspace{1mm}

\begin{tabularx}{\linewidth}{lX}
\textit{Input:}&
Graph $G=(V,E)$, set of links $E^*$, integers $\ka,\pe>0$, weight function $w:E^*\rightarrow \mathbb{Z}_+$, cost function $c:E^*\rightarrow
\mathbb{R}_+\cup\{\infty\}$.\\
\textit{Find:}& 
 minimum cost  link set $F\subseteq E^*$ such that $w(F)\le
\pe$ and $(V,E\cup F)$ is \ka-edge-connected.
\end{tabularx}
}}
\end{center}
\medskip

A problem instance is thus given by $(V,E,E^*,c,w,\ka,\pe)$. An $F\subseteq E^*$ for which $(V,E\cup F)$ is \ka-edge-connected is  called
an {\em augmenting link set}. If all weights are equal to one, we
simply refer to the problem as  {\em Minimum Cost Edge Connectivity
  Augmentation}. 

An edge between $x,y\in V$ will be denoted as $xy$.  For a link $f$,
we use $f=(x,y)$ if it is a link between $x$ and $y$; note
that there might be several links between the same nodes with
different weights.  We may ignore all links of weight $>\pe$. If for a
pair of nodes $u,v\in V$, there are two links $e$ and $f$ between $u$
and $v$ such that $c(e)\le c(f)$ and $w(e)\le w(f)$, then we may also
ignore the link $f$. It is convenient to assume that for every value
$1\le t\le \pe$ and every two nodes $u,v\in V$, there is exactly one
link $e$ between $u$ and $v$ with $w(e)=t$ (if there is no such link
in the input $E^*$, we can add one of cost $\infty$). This $e$ will be
referred to as the {\em $t$-link between $u$ and $v$}. With this
convention, we will assume that $E^*$ consists of exactly \pe{} copies
of ${V\choose 2}$: a $t$-link between any two nodes $u,v\in V$ for
every $1\le t\le \pe$. However, in the input links of infinite cost
should not be listed. (We avoid the discussion of exactly how the links are
represented in the input: as we express the size of the kernel in 
terms of the number of nodes/edges/links, the exact representation does not matter for our
results.)

As defined above, an optimal solution  to {\em Weighted Minimum Cost Edge Connectivity
  Augmentation} does not allow using the same link in $E^*$
twice. Motivated by the original (unweighted) problem, a natural
further restriction is to forbid using multiple links (of possibly
different weights) between the same two nodes $u$ and $v$. If the
input graph is already $(k-1)$-edge-connected, neither of these
restrictions makes a
difference, since given an augmenting edge set, deleting all but one
links from a parallel bundle is still an augmenting edge
set. In Section~\ref{sec:augm-arbitr-appendix} we investigate the
problem of augmenting an arbitrary (possibly disconnected) graph to 2-edge-connected, where
using parallel links may result in a cheaper solution. We first solve
here the problem with allowing multiple copies of the same link,
and in Section~\ref{sec:no-parallel}, we show how the problem can be
solved if parallel links are forbidden.

For a set $S\subseteq V$, by $G/S$ we mean the contraction of
$S$ to a single node $s$. That is, the node set of the contracted graph is $(V-S)\cup \{s\}$,
and every edge $uv$ with $u\notin S$, $v\in S$ is replaced by an edge $us$ (possibly creating parallel edges); edges inside $S$ are removed.
Note that $S$ is not assumed to be connected.
 We also contract the links to $E^*/S$ accordingly. If multiple $t$-links are
created between $s$ and another node, we keep only one with
minimum cost. 

 We say that two nodes $x$ and $y$ are {\em \ka{}-inseparable} if
there is no $x\bar y$-set $X$ with $d(X)<\ka$. By Menger's theorem, this is equivalent to the existence of \ka{} edge-disjoint paths between $x$ and $y$; this property can be tested in polynomial time by a max flow-min cut computation.
Let us say that the node set $S\subseteq V$ is {\em \ka{}-inseparable}
if any two nodes $x,y\in S$ are \ka{}-inseparable. It is easy to
verify that being \ka{}-inseparable is an equivalence relation.\footnote{To see
transitivity, observe that if $x$ and $y$ are \ka{}-inseparable and
$y$ and $z$ are \ka{}-inseparable, then a cut $X$ separating $x$ and
$z$ would either separate $x$ and $y$, or $y$ and $z$, a
contradiction.}  The maximal \ka{}-inseparable sets hence give a
partition of the node set $V$. The following proposition provides us
with a preprocessing step that can be used to simplify the instance:

\begin{proposition}\label{prop:contract}
For a problem instance $(V,E,E^*,c,w,\ka,\pe)$, let $S\subseteq V$
be a \ka{}-inseparable set of nodes.
Let us consider the instance obtained by the contraction of $S$.
Assume $\bar F\subseteq E^*/S$ is an
optimal solution to the contracted problem. Then the pre-image of
$\bar F$ in $E^*$ is an optimal solution to the original problem.
\end{proposition}
\fi\ifappendix
\begin{proof}
  We claim that for a link set $F\subseteq E^*$, $(V,E\cup F)$ is
  \ka{}-edge-connected if and only if adding the image $\bar F$ of $F$
  to the contracted graph is \ka{}-edge-connected.  It is
  straightforward that if $F$ is an augmenting link set, then so is
  $\bar F$.  Conversely, assume for a contradiction that $\bar F$ is
  an augmenting link set but $F$ is not.  This means that there exists
  a set $X\subseteq V$ with $d_E(X)+d_F(X)<\ka$.  Since $S$ is
  \ka{}-inseparable, either $S\subseteq X$ or $S\cap X=\emptyset$.
  This implies that under the contraction the image of $X$ will
  violate \ka{}-edge-connectivity in the augmented graph, a
  contradiction.  
\end{proof}
\fi\ifmain

Note that contracting a
\ka{}-inseparable set $S$ does not affect whether $x,y\not\in S$ are
\ka{}-inseparable.  Thus by Proposition~\ref{prop:contract}, we can
simplify the instance by contracting each class of the partition given
by the \ka{}-inseparable relation. Observe that after such a
contraction, there are no longer any \ka{}-inseparable pair of
nodes any more. Thus we may assume in our algorithms that every
pair of nodes can be separated by a cut of size smaller than $k$.

\fi

%% file: byone-v7.tex
\ifmain
\section{Augmenting edge connectivity by one}
\label{sec:augm-edge-conn}

\subsection{Metric instances}\label{sec:metric}
The following notions will be used for augmenting  edge-connectivity from 1 to 2 and
from 2 to 3. We formulate them here in a generic
way.
Assume the input graph is $(\ka-1)$-edge-connected.  Let $\cal D$
denote the set of all minimum cuts, represented by the node sets. That
is, $X\in {\cal D}$ if and only if $d(X)=\ka-1$. Note that, by the minimality of the cut, both $X$ and $V\setminus X$ induce connected graphs if $X\in {\cal D}$. For a link
$e=(u,v)\in E^*$, let us define ${\cal D}(e)\subseteq {\cal D}$ as the
subset of minimum cuts {\em covered} by $e$.  That is, $X\in {\cal D}$
is in ${\cal D}(e)$ if and only if $X$ is an $u\bar{v}$-set or a
$v\bar{u}$-set. Clearly, augmenting edge-connectivity by one is equivalent to
covering all the minimum cuts of the graph.
\begin{proposition}\label{prop:s-ec}
Assume $(V,E)$ is $(\ka{}-1)$-edge-connected. Then $(V,E\cup F)$ is
\ka{}-edge-connected if and only if $\cup_{e\in F} {\cal D}(e)={\cal D}$.
\end{proposition}
The following definition identifies the class of metric instances that plays a key role in our algorithm.
%
\begin{definition}
We say that the link $f$ is a {\em shadow} of link
$e$, if $w(f)\ge w(e)$ and ${\cal D}(f)\subseteq {\cal D}(e)$.
The instance $(V,E,E^*,c,w,\ka,\pe)$ is {\em metric}, if 
\begin{enumerate}[(i)]
\item $c(f)\le c(e)$  holds whenever the link $f$ is a shadow of link $e$.
\item Consider three links $e=(u,v)$, $f=(v,z)$ and $h=(u,z)$ with
  $w(h)\ge w(e)+w(f)$. Then $c(h)\le c(e)+c(f)$.
\end{enumerate}
\end{definition}

Whereas the input instance may not be metric, we can create its metric
completion with the following simple subroutine.
Let us call the inequalities in  {\em(i) shadow inequalities} and those in {\em (ii)  triangle inequalities}. Let us 
define  the {\em rank} of the inequality $c(f)\le c(e)$  to be $w(f)$, and the rank of $c(h)\le c(e)+c(f)$ to  be $w(h)$.
By {\em fixing} the triangle inequality $c(h)>c(e)+c(f)$, we mean decreasing the value of $c(h)$ to $c(e)+c(f)$.

\begin{figure}[t]
\fbox{\parbox{\textwidth}{
\begin{tabbing}
xxxxx \= xxx \= xxx \= xxx \= xxx \= \kill
\> \textbf{Subroutine} \textsc{Metric-Completion}$(c)$\\
\> \textbf{for} $t=1,2,\ldots,p$ \textbf{do} \\
\> \> \textbf{for} every 3 links  $e=(u,v)$, $f=(v,z)$, $h=(u,z)$ with
  $w(h)=t\ge w(e)+w(f)$ \textbf{do}\\
\> \> \> $c(h)\leftarrow \min\{c(h),c(e)+c(f)\}$\\
\> \> \textbf{for} every link $f$ with $w(f)=t$ \textbf{do}\\
\> \> \> $c(f)\leftarrow \min\{c(e): f\mbox{ is a shadow of }e\}$.
\end{tabbing}
}}\caption{The algorithm for computing the metric completion}\label{fig:completion-alg}
\end{figure}

The subroutine {\sc Metric-Completion$(c)$} (see
Figure~\ref{fig:completion-alg}) consists of $p$ iterations, one for
each $t=1,2,\ldots,p$. In the $t$'th iteration, first all triangle
inequalities of rank $t$ are taken in an arbitrary order, and the
violated ones are fixed. Then for every $t$-link $f$, we decrease
$c(f)$ to the minimum cost of links $e$ such that $f$ is a shadow of
$e$. Note that we perform these steps one after the other for every
violated inequality: in each step, we decrease the cost of a single
link $f$ only (this will be important in the analysis of the
algorithm).
The first part of iteration 1 is void as there are no rank 1 triangle inequalities. The subroutine can be implemented in polynomial time: the number of triangle inequalities is $O(p^3n^3)$, and they can be efficiently listed; furthermore, every link is the shadow of $O(pn^2)$ other ones.

\begin{lemma}\label{lem:metric-new}
Consider a problem instance $(V,E,E^*,c,w,\ka,\pe)$ with the graph $(V,E)$ being $(\ka-1)$-edge-connected.
{\sc Metric-Completion$(c)$} returns a metric cost function $\bar c$ with $\bar{c}(e)\le c(e)$ for every link $e\in E^*$. 
Moreover, if for a link set $\bar F\subseteq E^*$, graph $(V,E\cup \bar F)$ is
\ka{}-edge-connected, then there exists an  $F\subseteq E^*$
such that $(V,E\cup F)$ is \ka{}-edge-connected, $c(F)\le \bar{c}(\bar F)$,
and $w(F)\le w(\bar F)$. Consequently, an optimal solution for $\bar{c}$
provides an optimal solution for $c$.
\end{lemma}
\iffull \else 
The proof (deferred to the Appendix) proceeds by showing that after iteration $t$, all rank $t$ inequalities are satisfied and they remain satisfied later on.
\fi
\fi\ifappendix
\begin{proof}
Inequality $\bar c(e)\le c(e)$ clearly holds for all links since the algorithm only decreases the costs.
To verify the metric property, we prove that at the end of iteration $t$, all rank $t$ inequalities are satisfied.
This implies that the final cost function is metric, as the costs of the edges participating in rank $t$
inequalities are not modified during any later iteration.

Consider a triangle inequality with links $t=w(h)\ge w(e)+w(f)$. As $w(e),w(f)<t$, the costs of $e$ and $f$ are not modified in iteration $t$.
After fixing this inequality if necessary, we have $c(h)\le c(e)+c(f)$. In the second part of the iteration, $c(h)$ may  only decrease. Consequently, all triangle inequalities of rank $t$ must be valid at the end of iteration $t$.

Let $\tilde c$ denote the cost function at the end of the first part
of iteration $t$, after fixing all triangle inequalities. Using the
fact that the shadow relation is transitive, it is easy
to see that the values $c(f)$ after the second part of iteration $t$ equal
\begin{equation}\label{setshadow}
c(f)=\min\{\tilde c(e): f\mbox{ is a shadow of }e\}.
\end{equation}
Consider now two links $e$ and $f$ with $f$ being a shadow of $e$, and
let $t=w(f)\ge w(e)$. We have to show $c(f)\le c(e)$ at the end of
iteration $t$.
This is straightforward if $w(e)<t$: the new value of $c(f)$ is
defined as a minimum value taken over a set containing $c(e)$; $c(e)$ itself is not modified. Assume now $w(e)=t$.
Let $h$ be the link giving the minimum in (\ref{setshadow})
for the link $e$, that is, the  new value  is  $c(e)=\tilde c(h)$ with $e$
being the shadow of $h$. Again by the transitivity of the shadow
relation, $f$ is also a shadow of $h$, and consequently, $c(f)\le \tilde c(h)=c(e)$, as required.

For the second part of the lemma, it is enough to verify the statement for the case when  $\bar c$ arises by a single modification step from $c$ (i.e., fixing 
a triangle inequality or taking a minimum).
First, assume we fixed a triangle inequality $c(h)>c(e)+c(f)$ by setting $\bar c(h)=c(e)+c(f)$ and $\bar c(g)=c(g)$ for every $g\neq h$.
Consider an edge set $\bar F$ such that  $(V,E\cup \bar F)$ is
\ka{}-edge-connected. If $h\notin \bar F$, then $F=\bar F$ satisfies the conditions. If $h\in \bar F$, then let us set $F=(\bar F\setminus\{h\})\cup \{e,f\}$. We have $c(F)\le \bar c(F)$, $w(F)\le w(\bar F)$. Furthermore, every cut covered by $h$ must be covered by either $e$ or $h$, implying that 
$(V,E\cup F)$ is also \ka{}-edge-connected.

Next, assume $\bar c(f)=c(e)$ was set for a link 
$e$ such that $f$ is a shadow of $e$, and $\bar c(g)=c(g)$ for every
$g\neq f$. Now $F=(\bar F\setminus \{f\})\cup \{e\}$ clearly satisfies the conditions: recall that by the definition of shadows, ${\cal D}(f)\subseteq {\cal D}(e)$.
\end{proof}
\fi\ifmain
The proof also provides an efficient way for transforming an augmenting link set $\bar F$ to another $F$ as in the lemma.
For this, in every step of {\sc Metric-Completion$(c)$}  we have to keep track of the inequalities responsible for cost reductions. 

By Lemma~\ref{lem:metric-new}, we may restrict our attention to metric instances. In what follows, we show how to construct a kernel for metric instances for cases $\ka=2$ and $\ka=3$. (The case $\ka=2$ could be easily reduced to $\ka=3$, but we treat it separately as it is somewhat simpler and more intuitive.)
Section~\ref{sec:by-1} then shows how the case of general \ka{} can be reduced to either of these cases depending on the parity of \ka{}.

\subsection{Augmentation from 1 to 2}\label{sec:1-to-2}
In this section, we assume that the input graph $(V,E)$ is connected.
By Proposition~\ref{prop:contract}, we may assume that it is
 a tree: after contracting all the 2-inseparable sets, there are no two nodes with two edge-disjoint paths between them, implying that there is no cycle in the graph.

 The minimum cuts are given by the edges,
 that is, ${\cal D}$ is in one-to-one correspondence with $E$.
For a link $e$ between two nodes $u,v\in V$, let $P(e)=P(u,v)$ denote the unique path between
$u$ and $v$ in this tree. Then the link $f$ is a shadow of
the link $e$ if $P(f)\subseteq P(e)$ and $w(f)\ge w(e)$. Now
Proposition~\ref{prop:s-ec} simply amounts to the following.

\begin{proposition}\label{prop:2-ec}
Graph $(V,E\cup F)$ is 2-edge-connected if and only if $\cup_{e\in F} P(e)=E$.
\end{proposition}

Based on Lemma~\ref{lem:metric-new}, it
suffices to solve the problem assuming that the instance
$(V,E,E^*,c,w,2,\pe)$ is metric. The main observation is that in a
metric instance we only need to use  links that connect certain special
nodes, whose number we can bound by a function of $p$.

Let us refer to the leaves and nodes of degree at least 3 as {\em
  corner nodes}; let $R\subseteq V$ denote their set.  Every leaf in
the tree $(V,E)$ requires at least one incident edge in
$F$. If the number of leaves is greater than $2\pe$, 
we may conclude that the problem is infeasible. (Formally, in this case we may return the following kernel: 
a single edge as the input graph with an empty link set.)
  If there are at most $2\pe$ leaves, then $|R|\le
4\pe-2$, due to the following simple fact.
\begin{proposition}\label{prop:deg-3}
The number of nodes of degree at least 3 in a tree is at most the
number of leaves minus 2.
\end{proposition}
 
Based on the following theorem, we can obtain a kernel on at most
$4\pe-2$ nodes by contracting each path of degree-2 nodes to a single
edge. The number of links in the kernel will be $O(\pe^3)$: there are $O(\pe^2)$ possible edges and $\pe$ possible weights for each edge.

\begin{figure}
{\centering
\includegraphics[width=0.8\linewidth]{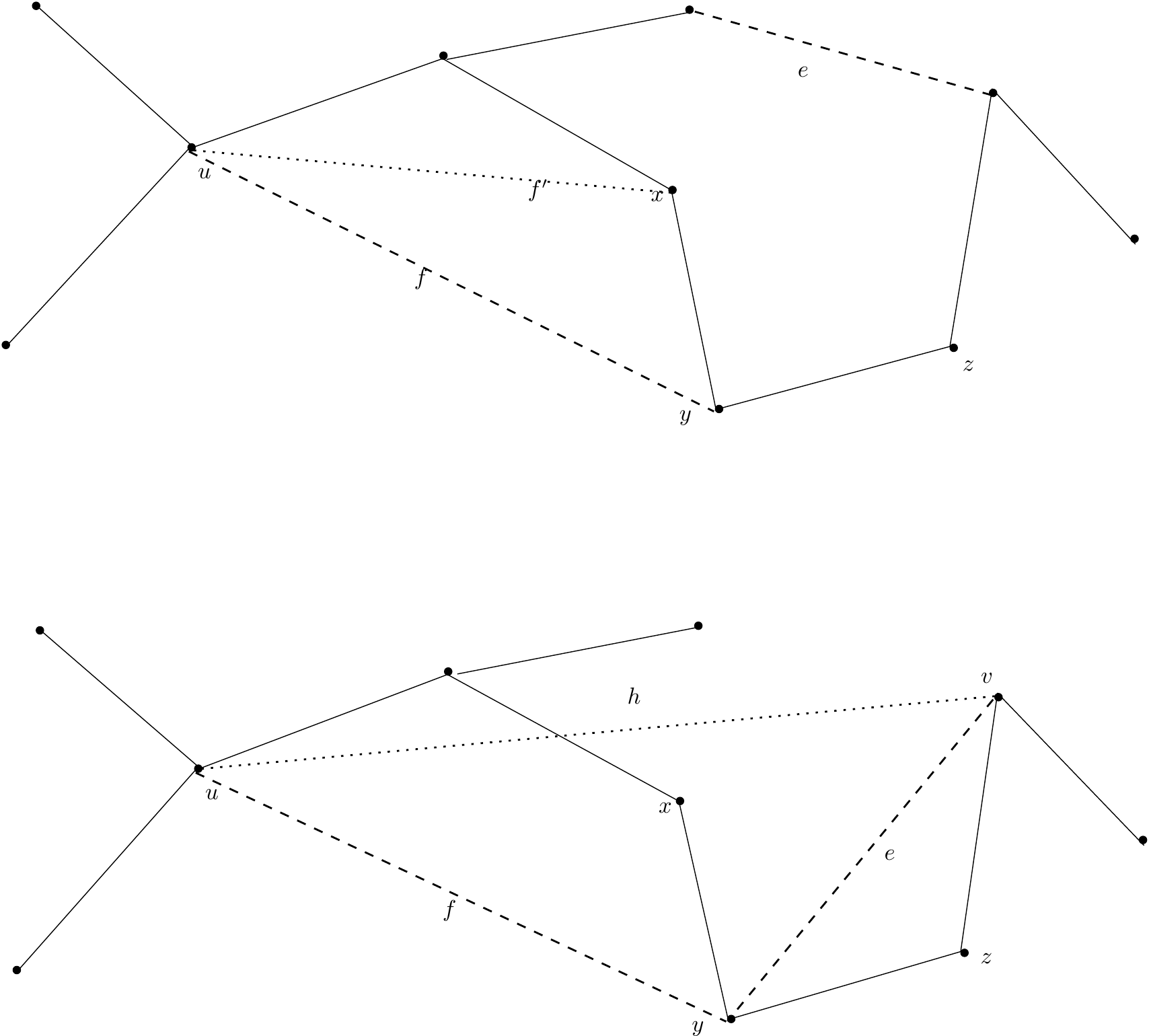}
\caption{Illustration of Cases I and II in the proof of Theorem~\ref{thm:degree-3}.}\label{fig:proof1}
}
\end{figure}

\begin{theorem}\label{thm:degree-3}
For a metric instance $(V,E,E^*,c,w,2,\pe)$, there exists an optimal
solution $F$ such that every edge in $F$ is only incident to corner nodes.
\end{theorem}
\iffull \else 
The proof (see Appendix) analyses an optimal solution with the total number of links minimal, and subject to this, the total
length of the paths in the tree between the endpoints of the links minimal. Such an optimal solution may contain no links incident to degree 2 nodes.
\fi
\fi\ifappendix
\begin{proof}
For every link $f$, let $\ell(f)=|P(f)|$ denote the length of the
path in the tree between its endpoints.
Consider an optimal solution $F$ such that $|F|$ is minimal, and
subject to this, $\ell(F)=\sum_{e\in F}\ell(f)$ is minimal.
We show that no link in this set $F$ can be incident to a degree 2 node.

For a contradiction, assume that $f=(u,y)\in F$ has an endnode $y$ having degree
2 in $E$; let $x$ and $z$ denote the two neighbors of $y$, with $xy\in
P(f)$. Since $(V,E\cup F)$ is 2-edge-connected, there must be a link
$e\in F$ with $yz\in P(e)$. We distinguish two cases, as illustrated
in Figure~\ref{fig:proof1}.

{\bf Case I.}  $xy\in P(e)$. In this case, we may replace the link
$f=(u,y)$ by a link $f'=(u,x)$ with $w(f')=w(f)$. By property (i) of metric instances, 
we have  $c(f')\le c(f)$
as $f'$ is a shadow of
$f$. By Proposition~\ref{prop:2-ec}, $(V,E\cup F') $ is still
2-edge-connected for the resulting solution $F'$, yet $c(F')\le c(F)$
and $\ell(F')<\ell(F)$, a contradiction to the choice of $F$.

{\bf Case II.} $xy\notin P(e)$. This is only possible if $e$ is
incident to $y$, say $e=(y,v)$. For $t=w(f)+w(e)$, consider the $t$-link
$h$ between $u$ and $v$. By property (ii), $c(h)\le
c(f)+c(e)$. Furthermore, $P(h)=P(f)\cup P(e)$. For the resulting
solution $F'$, graph
$(V,E\cup F')$ is 2-edge-connected, $c(F')\le c(F)$ and $|F'|<|F|$, a
contradiction again.
\end{proof}
\fi\ifmain

\subsection{Augmentation from 2 to 3}\label{sec:2-to-3}
In this section, we assume that the input graph is 2-edge-connected but
not 3-edge-connected.  Let us call a 2-edge-connected graph $G=(V,E)$
a {\em cactus}, if every edge belongs to exactly one circuit. This is
equivalent to saying that every block (maximal induced
2-node-connected subgraph) is a circuit (possibly of length 2, using
two parallel edges). Figure~\ref{fig:cactus} gives an example of a
cactus.

\begin{figure}
{\centering
\includegraphics[width=0.5\linewidth]{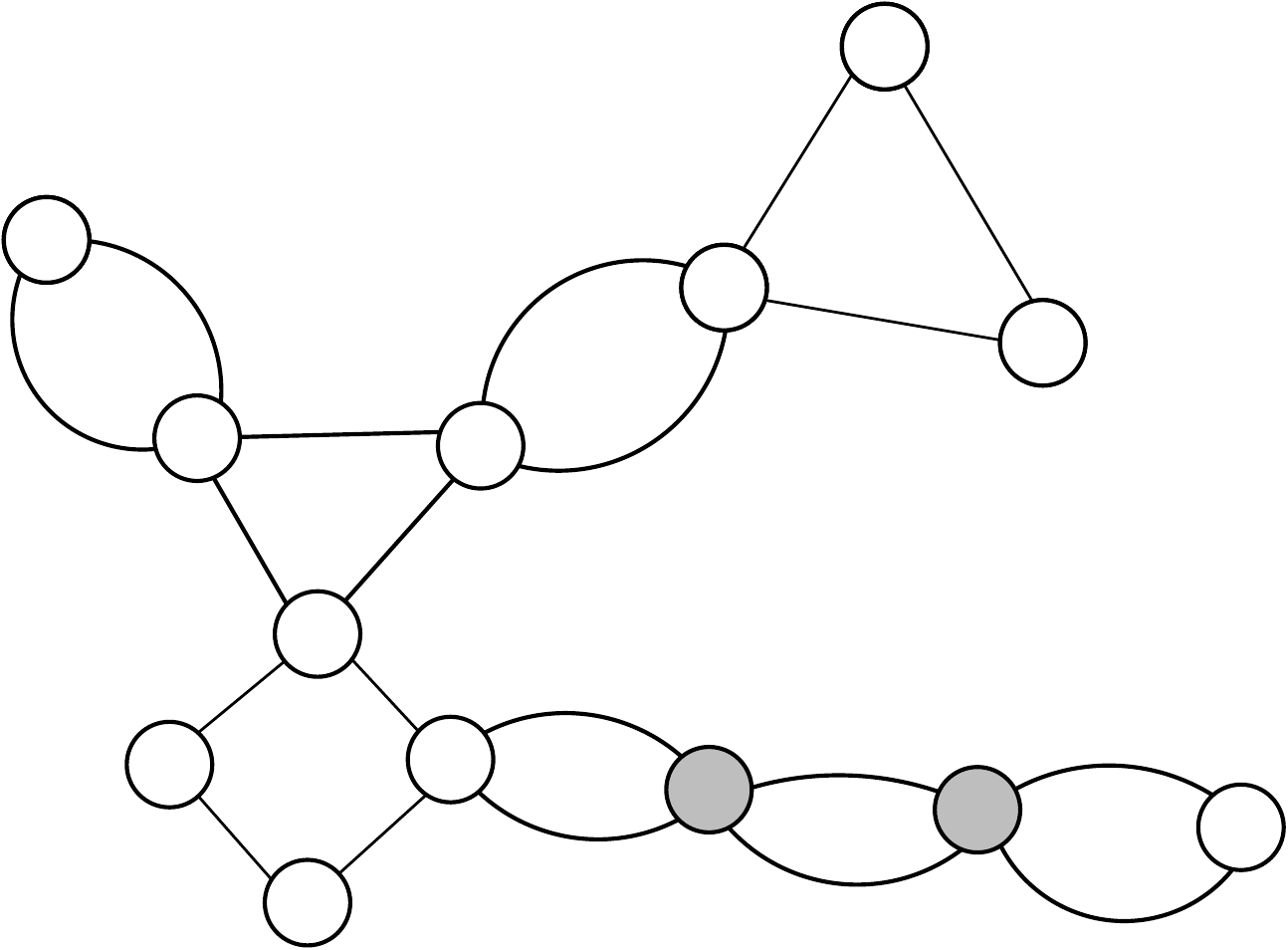}
\caption{A cactus graph. The shaded nodes are in the set $T$.}\label{fig:cactus}
}
\end{figure}

By Proposition~\ref{prop:contract}, we may assume that every
$3$-inseparable set in $G$ is a singleton, that is, there are no two nodes in
the graph connected by 3 edge-disjoint paths. 

\begin{proposition}\label{prop:cactus}
Assume that $G=(V,E)$ is a 2-edge-connected graph such that every
3-inseparable set is a singleton. Then $G$ is a cactus.
\end{proposition}
\fi\ifappendix
\begin{proof}
  By 2-edge-connectivity, every edge must be contained in at least one
  circuit.  For a contradiction, assume there is an edge $e$ contained in two different circuits 
   $C_1$ and $C_2$. Pick an edge $f\in
   C_1\setminus C_2$, and take the maximal path $P$ in $C_1$
   containing $f$ such that the nodes incident to both $P$ and $C_2$
   are precisely the endpoints of $P$, say $x$ and $y$.
   The edge $e\in C_1\cap C_2$ guarantees the existence of such a
   path, that is, $x\neq y$.
Now there
  are three edge-disjoint paths connecting $x$ and $y$: $P$ and the
  two $x-y$ paths contained on $C_2$. This contradicts our assumption.  
\end{proof}
\fi\ifmain

In the rest of the section, we assume that $G=(V,E)$ is a cactus.
The set of minimum cuts $\cal D$ corresponds to 
arbitrary pairs of 2 edges on the same circuit. 
 We say that the node
$b$ {\em separates} the nodes $a$ and $c$, if every path between $a$ and $c$
must traverse $b$ (we allow $a=b$ or $b=c$).

\begin{proposition}\label{prop:3-shadow}
Consider links $e=(u,v)$ and $f=(x,y)$
with $w(f)\ge w(e)$. Then $f$ is a shadow of $e$ if and only if both
$x$ and $y$ separate $u$ and $v$.
\end{proposition}
\fi\ifappendix
\begin{proof}
To see sufficiency, assume that both $x$ and $y$ separate $u$ and $v$, and
consider an $x\bar{y}$-set $X\in {\cal D}(f)$. 
We have to show that $X\in {\cal D}(e)$, that is, one of $u$ and $v$ is in $X$ and the
other in $V\setminus X$. Indeed, assume for a contradiction that $u,v\in X$. Since
$X$ is connected, it contains a path between $u$ and $v$ avoiding $y$,
a contradiction.

For necessity, assume $x$ does not separate $u$ and $v$, that is,
there exists a path $Q$ between $u$ and $v$ not containing $x$.
Pick two edges incident to $x$ that are contained in the same cycle.
They correspond to a minimum cut $X\in {\cal D}(f)$ (they are the two edges between $X$ and $V-X$).
The path $Q$ is either entirely contained in $X$ or in $V-X$ (as it cannot traverse the edges incident to $x$), and therefore  $e=(u,v)$ cannot cover $X$.
This contradicts ${\cal D}(f)\subseteq {\cal D}(e)$.
\end{proof}
\fi\ifmain

Again by Lemma~\ref{lem:metric-new},
we may restrict our attention to metric instances. 
Let us call a circuit of length 2 a {\em 2-circuit} (that is, a set of two parallel edges between two nodes).
Let $R_1$ denote
the set of nodes of degree 2, or equivalently, the set of nodes
incident to exactly one circuit. Let $R_2$ denote the set of nodes
incident to at least 3 circuits, or at least two circuits not both
2-circuits. Let $R=R_1\cup R_2$ and let $T=V\setminus R$ denote the
set of remaining nodes, that is, the set of nodes that are incident to
precisely two circuits, both 2-circuits (see Figure~\ref{fig:cactus}). The elements of $R$ will be again called {\em corner nodes}.
We can give the following
simple bound:

\begin{proposition}\label{prop:r1-r2}
$|R_2|\le 4|R_1|-8$.
\end{proposition}
\fi\ifappendix
\begin{proof}
The proof is by induction on $|V|$. If all circuits in $G$ are 2-circuits,
that is, $G$ is created by duplicating every edge of a tree, $R_1$ corresponding to the leaves and $R_2$ to the branching nodes.
The claim follows by Proposition~\ref{prop:deg-3},
 as $|R_1|\ge 2$.
Assume now $G$ has at least one circuit $C$ of length 
$r\ge 3$, and has $t\le r$ nodes incident to other circuits.
Consider the graph after removing the edges of $C$ and the
$r-t$ isolated nodes. We obtain $t$ cacti; let $a_i$ and $b_i$ denote
the corresponding $|R_2|$ and $|R_1|$ values for $i=1,\ldots, t$.
By induction, $b_i\le 4a_i-8$ holds for each of them, giving
\begin{equation}
\sum_{i=1}^t b_i\le \sum_{i=1}^t(4a_i-8)=4\sum_{i=1}^t(a_i-1)-4t.\label{eq:induction}
\end{equation}
Observe that $|R_2|\le \sum_{i=1}^t b_i+t$, since the only nodes of $R_2$ that are possibly not accounted for in any of the smaller cacti are the $t$ nodes where these cacti are incident to $C$.  Also, $|R_1|\ge \sum_{i=1}^t (a_i-1)+r-t$, since
we remove at most one node of degree 2 from each component and add $r-t$ new ones.
Adding up the inequalities we obtain
\begin{align*}
|R_2|\le \sum_{i=1}^t b_i +t\le 4\sum_{i=1}^t
(a_i-1)-3t\\
\le 4(|R_1|+t-r)-3t=4|R_1|+t-4r\le 4|R_1|-8
\end{align*}
The second inequality holds by (\ref{eq:induction}), and the last one uses $8\le
4r-t$ that is valid since $t\le r$ and $r\ge 3$.
\end{proof}
\fi\ifmain

Observe that every
node in $R_1$ forms a singleton minimum cut.
Hence if $|R_1|> 2\pe$, we may conclude infeasibility. Otherwise,
Proposition~\ref{prop:r1-r2} gives $|R|\le 10\pe-8$.

We prove  the analogue of Theorem~\ref{thm:degree-3}: we show that it
is sufficient to consider only links incident to $R$%
 \iffull \else 
(see Theorem~\ref{thm:degree-3-3})%
\fi.  It follows that
we can obtain a kernel on at most $10\pe-8$ nodes by replacing every
path consisting of 2-circuits by a single 2-circuit. The number of
links in the kernel will again be $O(\pe^3)$. 
\fi\ifappendix

\begin{theorem}\label{thm:degree-3-3}
For a metric instance $(V,E,E^*,c,w,3,\pe)$, there exists an optimal
solution $F$ such that every edge in $F$ is only incident to corner nodes.
\end{theorem}
\fi\ifappendix
\begin{proof}
The proof goes along the same lines as that of
Theorem~\ref{thm:degree-3}.
For every link $f$, let $\ell(f)=|{\cal D}(f)|$.
Consider an optimal solution $F$ such that $|F|$ is minimal, and
subject to this, $\ell(F)=\sum_{e\in F}\ell(f)$ is minimal.
We show that no link in this set $F$ can be incident to a node in $T$.

For a contradiction, assume $f=(u,y)\in F$ has an endnode $y\in T$.
Node $y$ is incident to two 2-circuits; let us denote these by $C_x$ and $C_z$, with $C_x$ consisting of two parallel edges between $x$ and $y$ and $C_z$ between $y$ and $z$.
 Clearly, $f$  covers exactly
one of the corresponding two cuts. W.l.o.g. assume that the cut
corresponding to  $C_x$ is in  ${\cal D}(f)$; note that this implies that $x$ separates $u$ and $y$.
Since $(V,E\cup F)$ is 3-edge-connected, there must be a link
$e\in F$ such that the cut corresponding to  $C_z$ is in ${\cal D}(e)$.
The two cases whether the cut corresponding  to $C_x$ is in ${\cal
  D}(e)$ lead to contradictions the same way as in the proof of Theorem~\ref{thm:degree-3}, using Proposition~\ref{prop:3-shadow}.
%
%
\end{proof}
\fi\ifmain

\subsection{Augmenting edge-connectivity for higher values}\label{sec:by-1}
In this section, we assume that the input graph $G=(V,E)$ is already $(\ka{}-1)$-connected, where
\ka{} is the connectivity target. We show that for even or odd \ka{}, the
problem can be reduced to the $k=2$ or the $k=3$ case, respectively.

Assume first that \ka{} is even. We use the following simple structure theorem, which is based on the observation that if the minimum cut value in a graph is odd, then the family of minimum cuts is cross-free.
\begin{theorem}[{\cite[Thm 7.1.2]{frankkonyv}}]\label{thm:tree-structure}
Assume that the minimum cut value $\ka-1$ in the graph $G=(V,E)$ is
odd. Then there exists a tree $H=(U,L)$ along with a map
$\varphi:V\rightarrow U$ such that the min-cuts of $G$ and the edges
of $H$ are in one-to-one correspondence: for every edge $e\in L$, the
pre-images of the two components of $H-e$ are the sides of the
corresponding min-cut, and every minimum cut can be obtained this way.
\end{theorem}
Note that Theorem~\ref{thm:tree-structure} {\em does not} say that $G$ is a somehow a tree with duplicated edges: it is possible $x$ and $y$ are adjacent in $G$ even if $\phi(x)$ and $\phi(y)$ are not adjacent in the tree $H$ (see Figure~\ref{fig:exampletree}).

\begin{figure}
{\centering
\includegraphics[width=0.8\linewidth]{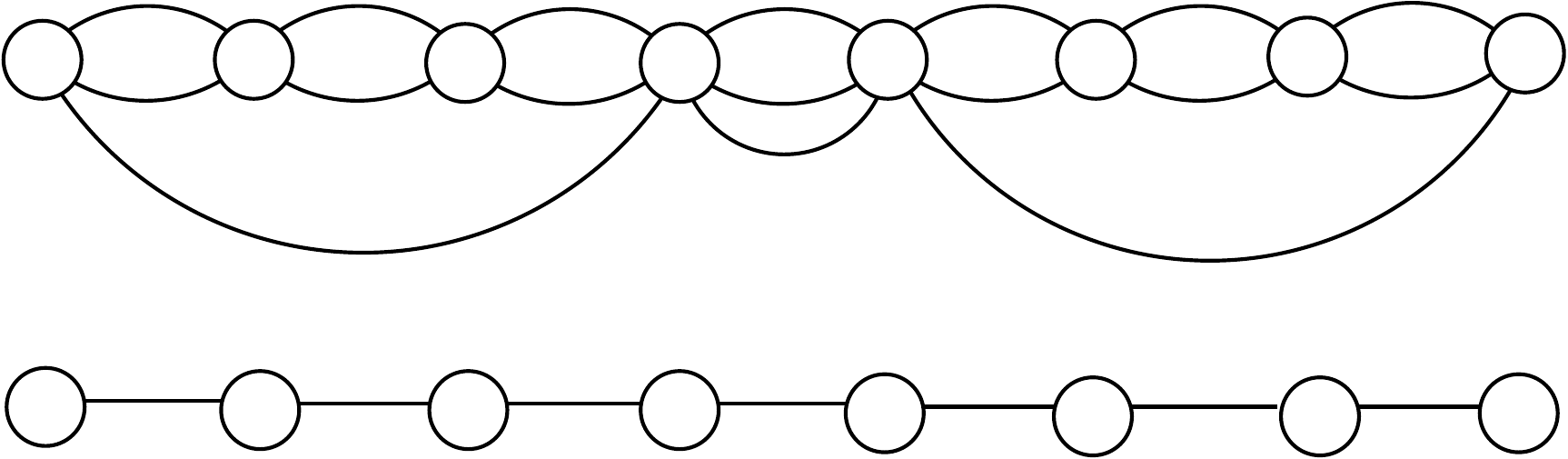}
\caption{Illustration of Theorem~\ref{thm:tree-structure} for
  $\ka=4$. The above graph is mapped to the path below with a
  bijection between the nodes. \label{fig:exampletree}}
}
\end{figure}

For even $\ka-1$, the following theorem shows that the minimum cuts can be
represented by a cactus. Note that the theorem also holds for odd $\ka-1$; however, in this case it is easy to see that the cactus arises from a tree by doubling all edges and hence obtaining Theorem~\ref{thm:tree-structure}.
\begin{theorem}[Dinits, Karzanov, Lomonosov \cite{dinits}, {\cite[Thm 7.1.8]{frankkonyv}}]\label{thm:cactus-structure}
Consider a loopless graph $G=(V,E)$ with minimum cut value $\ka-1$. Then
there exists a cactus $H=(U,L)$ along with a map $\varphi:V\rightarrow
U$ such that the min-cuts of $G$ and the edges of $H$ are in
one-to-one correspondence. That is, for every minimum cut $X\subseteq U$ of
$H$, $\varphi^{-1}(X)$ is a minimum cut in $G$, and
every minimum cut in $G$ can be obtained in this form.
\end{theorem}

Observe that if $G$ does not contain $k$-inseparable pairs (e.g., it
was obtained by contracting all the maximal $k$-inseparable sets),
then $\varphi$ in Theorems~\ref{thm:tree-structure} and
\ref{thm:cactus-structure} is one-to-one: $\varphi(x)=\varphi(y)$ would mean that
there is no minimum cut separating $x$ and $y$.  Therefore, in
this case Theorems~\ref{thm:tree-structure} and
\ref{thm:cactus-structure} imply that we can replace the graph with a
tree or cactus graph $H$ in a way that the minimum cuts are
preserved. Note that the {\em value} of the minimum cut does change:
it becomes 1 (if $H$ is a tree) or 2 (if $H$ is a cactus), but
$X\subseteq V$ is a minimum cut in $G$ if and only if it is a minimum
cut in $H$. The proof of the above theorems also give rise to polynomial time algorithms that find the tree or cactus representations efficiently. Let us summarize the above arguments.
\begin{lemma}\label{lem:contract-insep}
  Let $G=(V,E)$ be a $(k-1)$-edge-connected graph containing no
  $k$-inseparable pairs. Then in polynomial time, one can construct a
  graph $H=(V,L)$ on the same node set having exactly the same set
  of minimum cuts such that
\begin{enumerate}
\item if $k$ is even, then $H$ is a tree (hence the minimum cuts are of size 1), and
\item if $k$ is odd, then $H$ is a cactus (hence the minimum cuts are of size 2).
\end{enumerate}
\end{lemma}

Now we are ready to show that if $G$ is $(k-1)$-edge-connected, then a
kernel containing $O(p)$ nodes, $O(p)$ edges,  and $O(p^3)$ links is possible for every
$k$.  First, we contract every maximal $k$-inseparable set; if multiple
links are created between two nodes with the same weight, let us only
keep one with minimum cost. By Proposition~\ref{prop:contract}, this
does not change the problem. Then we can apply
Lemma~\ref{lem:contract-insep} to obtain an equivalent problem on
graph $H$ having a specific structure.  If $k$ is even, then covering
the $(k-1)$-cuts of $G$ is equivalent to covering the $1$-cuts of the
tree $H$, that is, augmenting the connectivity of $G$ to $k$ is
equivalent to augmenting the connectivity of $H$ to $2$. Therefore, we
can use the algorithm described in Section~\ref{sec:1-to-2} to obtain
a kernel.  If $k$ is odd, then covering the $(k-1)$-cuts of $G$ is
equivalent to covering the $2$-cuts of the cactus $H$, that is,
augmenting the connectivity of $G$ to $k$ is equivalent to augmenting
the connectivity of $H$ to $3$. In this case, Section~\ref{sec:2-to-3}
gives a kernel.

\subsection{Decreasing the size of the cost}\label{sec:size-decrease}
We have shown that for arbitrary instance $(V,E,E^*,c,w,\ka,\pe)$, if
$(V,E)$ is $(\ka{}-1)$-edge-connected, then there exists a kernel on $O(\pe)$
nodes and $O(\pe^3)$ links. However, the costs of the links in this
kernel can be arbitrary rational numbers (assuming the input contained
rational entries).

We show that the technique of Frank and Tardos \cite{franktardos} is
applicable to replace the cost by integers whose size is polynomial in
\pe{} and the instance remains equivalent to the original one.

\begin{theorem}[\cite{franktardos}]\label{th:franktardos}
Let us be given a rational vector $c=(c_1,\ldots,c_n)$ and an integer
$N$.
Then there exists an integral vector $\bar c=(\bar c_1,\ldots,\bar
c_n)$ such that $||\bar c||_\infty\le 2^{4n^3}N^{n(n+2)}$ and
$\mbox{sign}(c\cdot b)=\mbox{sign}(\bar c\cdot b)$, where $b$ is an
arbitrary integer vector with $||b||_1\le N-1$. Such a vector $\bar c$
can be constructed in polynomial time.
\end{theorem}

In our setting, $n=O(\pe^3)$ is the length of the vector.  We want to modify the cost function $c$ to obtain a new cost function $\bar c$ with the following property: for arbitrary two sets of links $F,F'$ with
$|F|,|F'|\le\pe$, we have $c(F)<c(F')$ if and only if
$\bar c(F)<\bar c(F')$. This can be guaranteed by requiring that $\mbox{sign}(c\cdot b)=\mbox{sign}(\bar c\cdot b)$ for every vector $b$ containing at most $2\pe$ nonzero coordinates, all of them being $1$ or $-1$. Thus it is sufficient to consider vectors 
$b$ with $||b||_1\le 2\pe$, giving $N=2\pe+1$.  Therefore Theorem~\ref{th:franktardos}
provides a guarantee $||\bar c||_\infty\le
2^{O(\pe^6)}(2\pe+1)^{O(\pe^6)}$, meaning that each entry of $\bar c$
can be described by $O(\pe^6\log \pe)$ bits.  An optimal solution for
the cost vector $\bar c$ will be optimal for the original cost $c$.
This completes the proof of Theorem~\ref{th:main-weighted}.

\begin{remark}
The above construction works for {\em Weighted Minimum Cost Edge
Connectivity Augmentation} defined as an optimization problem. However,
parametrized complexity theory traditionally addresses decision
problems. The corresponding decision problem further includes a value
$\alpha\in \mathbb{R}$ in the input, and requires to decide whether
there exists an augmenting edge set of weight at most $p$ and cost at
most $\alpha$. For this setting, we can apply the Frank-Tardos
algorithm for the vector $(c,\alpha)$ instead of $c$; this gives the
same complexity bound  $O(\pe^6\log \pe)$. 
\end{remark}

\fi\ifappendix
\input{unweighted-v7}

\subsection{Node-connectivity augmentation}\label{sec:node}
Consider an instance $(V,E,E^*,c,w,2,p)$ of {\em Weighted Minimum Cost
   Node-Con\-nectivity Augmentation from 1 to 2}. We reduce it to an
instance of {\em Weighted Minimum Cost Edge-Connectivity Augmentation
  from 1 to 2} via a simple and standard construction.

Let $N\subseteq V$ denote the set of cut nodes in $G=(V,E)$. Let us
perform the following operation for every $v\in N$ (illustrated on Figure~\ref{fig:split}).  Let $V_1,\ldots,
V_r$ denote the node sets of the connected components of $G-v$;
$r\ge 2$ as $v$ is a cut node.  Let us add $r$ new nodes
$v_1,v_2,\ldots,v_r$, connected to $v$. Replace every edge $uv\in E$
with $u\in V_i$ by $uv_i$ and similarly every link $(u,v)$ with $u\in
V_i$ by a link $(u,v_i)$ of the same cost and weight. Note that there
are exactly $r$ edges and no links incident to $v$ after this
operation. Let us call the $vv_i$ edges {\em special edges}.
\begin{figure}
{\centering
\includegraphics[width=0.9\linewidth]{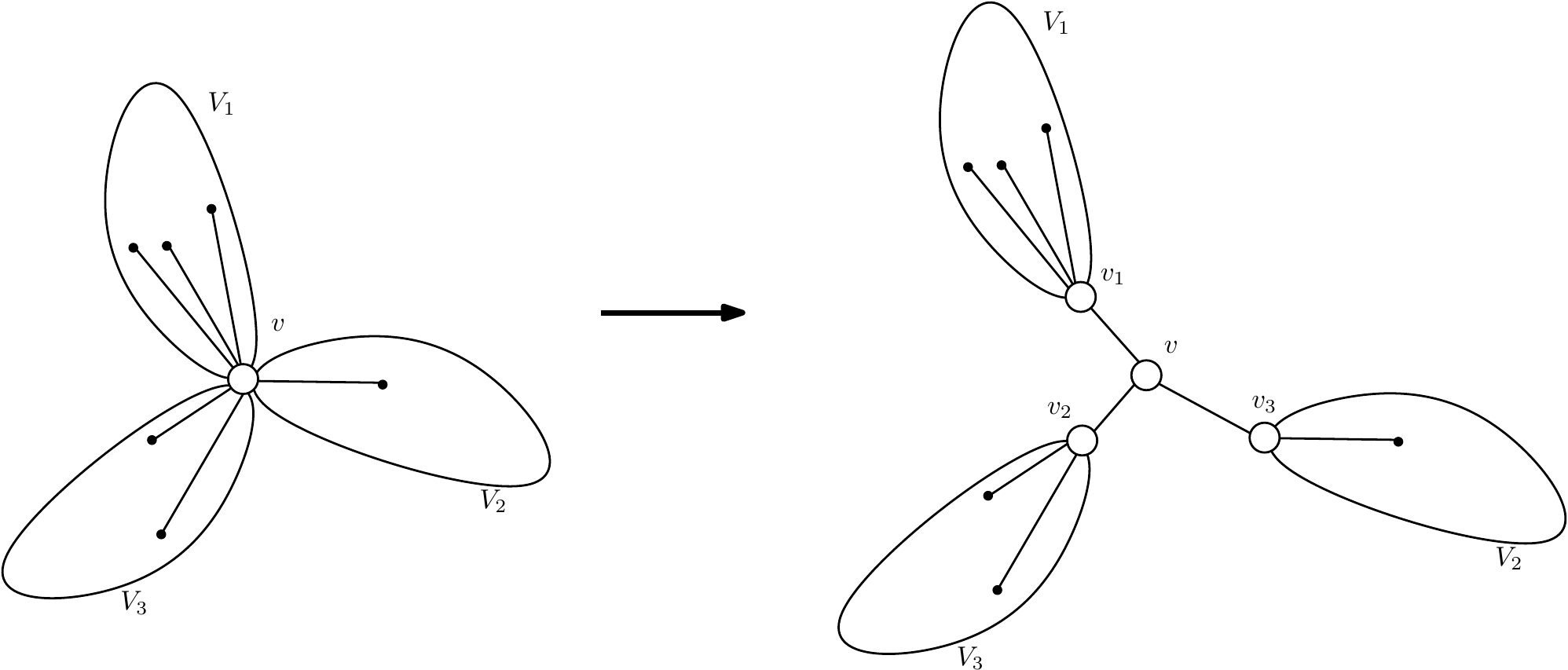}
\caption{The node splitting operation.}\label{fig:split}
}
\end{figure}

Let $G'=(V',E')$ denote the resulting graph after performing this for every $v\in N$. For a link set $F$, let $\varphi(F)$ denote its image
after these operations. The following lemma shows the reduction to the {\em Weighted Minimum Cost Edge-Connectivity Augmentation from 1 to 2} problem.

\begin{lemma}\label{lem:2-node-conn}
Graph $(V,E\cup F)$ is 2-node-connected if and only if $(V',E'\cup \varphi(F))$ is 2-edge-connected.
\end{lemma}
\begin{proof}
  Consider first a link set $F$ such that $(V,E\cup F)$ is
  2-node-connected. Assume that there is a cut edge in $(V',E'\cup
  \varphi(F))$.  If it is an edge $e\in E'$ that is an image of an
  original edge from $E$, then it is easy to verify that $e$ must also be a
  cut edge in $(V,E\cup F)$. If the cut edge is some edge $vv_i$ added in the
  construction, then $V_i$ is disconnected from the rest
  of the graph in $(V,E\cup F)-v$. The converse direction follows by
  the same argument.
\end{proof}

It is left to prove that a kernel $(V'',E'')$ for the
edge-connectivity augmentation problem can be transformed to a kernel
of the node-connectivity augmentation problem. Graph $(V'',E'')$ was
obtained by first contracting the maximal 2-inseparable sets, then contracting all paths of degree 2 nodes in the resulting tree.
In the first step, no special edges can be contracted, since $v$ and $v_i$ are not 2-inseparable. Also, if $v$ was an original cut node, then,
after the transformation, no link is incident to $v$. Is is not difficult to see that contracting all special edges in $(V'',E'')$ gives 
an equivalent node-connectivity augmentation problem.

\fi


%% file: unweighted-v7.tex
\subsection{Unweighted problems (Proof of Theorem~\ref{th:main})}\label{sec:main-proof}
In this section we show how Theorem~\ref{th:main} for unweighted instances can be deduced from 
Theorem~\ref{th:main-weighted}. 

Consider an instance of {\em Minimum Cost Edge-Connectivity Augmentation by One}: let $G=(V,E)$ be a $(k-1)$-edge-connected  and $E^*_0$ be a set of (unweighted) links with cost vector $c$.
We may take it as an instance of {\em Weighted Minimum Cost Edge-Connectivity Augmentation by One}, setting the weights of all links 1.
Theorem~\ref{th:main-weighted} then returns a kernel with $O(p)$ nodes and $O(p^3)$ links.

The first step in constructing the kernel was Lemma~\ref{lem:contract-insep}, which obtained an equivalent problem instance with the input $G=(V,E)$ being a tree or a cactus, and the connectivity target $k=2$ or $k=3$, respectively.
Let $R\subseteq V$ denote the set of corner nodes as in Sections~\ref{sec:1-to-2} and \ref{sec:2-to-3}, respectively; let $T=V\setminus R$.
 The kernel graph is obtained from $G$ by contracting all paths of degree 2 nodes  to single edges in trees, and all paths of 2-circuits to single 2-circuits in cacti. This was possible because in the metric closure, we can always find an optimal solution using links between corner nodes only (Theorems~\ref{thm:degree-3} and \ref{thm:degree-3-3}). 

Let $c$ denote the original cost function and $\bar c$ the one obtained by {\em Metric-Closure}$(c)$. Consider now a link in the kernel;
it corresponds to a link $f$ in the metric closure in $G$. Let us say that a set of (unweighted) links $A\subseteq E^*_0$ {\em emulates} a link $f$ in the metric closure, if
\begin{itemize}
\item $|A|\le w(f)$,
\item  $\sum_{e\in A(f)} c(f)\le \bar c(f)$, and
\item $\cup_{e\in A(f)}{\cal D}(e)\supseteq {\cal D}(f)$.
\end{itemize}
It is easy to verify that for every link $f$ in the metric closure there exists a set $A(f)$ emulating it (see also the proof of Lemma~\ref{lem:metric-new}).
 In every optimal solution, we may replace $f$ by the set of links $A(f)$ maintaining optimality. Then $|A(f)|\le \pe$ follows from  $w(f)\le \pe$.

We have shown that the $O(p^3)$ links in the weighted kernel may be replaced by $O(p^4)$ original links. This also increases the number of nodes and edges in the kernel, as we must keep all nodes in $T$ incident to these links. The bound $O(p^8\log p)$ on the bit sizes easily follows as in Section~\ref{sec:size-decrease}.

%% file: arbitrary-v7.tex
\ifmain

\fi
\ifappendix\section{Augmenting arbitrary graphs to 2-edge-connectivity}\label{sec:augm-arbitr-appendix}

In this section, we allow an arbitrary input graph; by
Proposition~\ref{prop:contract}, we may assume that $G=(V,E)$ is a
forest with $r>1$ components, denoted by $(V_1,E_1)$, $(V_2,E_2)$,
$\ldots,$ $(V_r,E_r)$ (we also
consider the isolated nodes as separate components, hence
$V=\cup_{i=1}^r V_i$).
There are two types of links in $E^*$: $e=(u,v)$ is an {\em internal
  link} if $u$ and $v$ are in the same component and {\em external
  link} otherwise.

In the following, we allow adding multiple copies of the same link. Doing
this can make sense if the link connects two different components:
then the two copies of the same link provides 2-edge-connectivity
between the two components. However, the problem was originally
defined such that multiple copies of the same link cannot be taken
into the solution.  In Section~\ref{sec:no-parallel}, we describe a
clean reduction how to enforce that there can be only one copy of each
link in the solution.

On a high level, we follow the same strategy as in
Section~\ref{sec:1-to-2}: we define an appropriate notion of metric
instances, and show that every input instance can be reduced
efficiently to an equivalent metric one. However, this reduction is
more involved than the reduction for connected inputs.  We are only able to establish
a fixed-parameter algorithm for metric instances, but we are unable to
construct a polynomial kernel.  In Section~\ref{sec:completion}, we will show how
to reduce the problem from arbitrary instances to metric ones.  Then
in Section~\ref{sec:metric-fpt}, we exhibit the FPT algorithm for
metric instances.
The following propositions and definitions are needed for the definition of metric instances.
%

\begin{proposition}\label{prop:2-ec-gen}
Graph $(V,E\cup F)$ is 2-edge-connected if and only if it is connected and
for every edge $e\in E\cup F$, there is a circuit in $E\cup F$
containing it.
\end{proposition}

 As before, if $f$ is an internal link connecting two nodes
in $V_i$, let $P(f)$ denote the unique path between the endpoints of
$f$ in $E_i$.
We also say that the node $y$  {\em lies between} the nodes $x$ and
$z$ if $x,y$ and $z$ are in the same component, and $y$ is contained
in the unique path between $x$ and $z$ in this component ($y=x$ or
$y=z$ is possible).
Furthermore, the edge $uv\in E$ is between $x$ and $z$, if it lies on the unique path between $x$ and $y$
in $E$ (equivalently, both $u$ and $v$ are between $x$ and $z$).
We will use the following fundamental property of circuits in a graph, the so-called strong circuit axiom in matroid theory.

\begin{proposition}\label{prop:strong-circuit}
Let $C$ and $C'$ be two circuits in a graph with $f\in C\cap C'$ and $g\in C\setminus C'$. Then there exists a circuit $C''$ with 
$C''\subseteq C\cup C'$, $g\in C''$ and $f\notin C''$.
\end{proposition}

Assume that the graph $(V,E\cup F)$ is 2-edge-connected. By
Proposition~\ref{prop:2-ec-gen}, for every $e\in E\cup F$ there exists
a circuit in $E\cup F$ containing $e$. Let $C(e)$ denote such a
circuit containing $e$ with $|C(e)\cap F|$ minimum (that is,
$C(e)$ contains a minimum number of links); if there are more than one, pick such a circuit arbitrarily.

\begin{proposition}\label{prop:intersect}
For every $e\in E\cup F$, consider the circuit $C(e)$. Then for every
$1\le i\le r$, if $C(e)$ intersects $(V_i,E_i)$, then the intersection is a
path (possibly a single node), and $C(e)$ contains either a single internal link between two nodes in $V_i$ or exactly  two external
links incident to $V_i$.
\end{proposition}
\begin{proof}
First, assume $C(e)$ contains an internal link $f$ incident to $V_i$. If $e=f$ is itself
this internal link, then $C(e)$ must consist of $e$ and the unique
path $P(e)$ in $E_i$ connecting the two endpoints of $e$. Indeed, this circuit
contains the minimum number of links (one), and furthermore there is no other
circuit in $E\cup \{e\}$; hence $C(e)$ is uniquely defined in this
case. On the other hand, if $e\neq f$, then let us apply
Proposition~\ref{prop:strong-circuit} to $C(e)$ and $C(f)$ (note that
$C(f)$ consists of $f$ and a path in $E_i$ by the above
argument). This gives a circuit $C\subseteq C(e)\cup C(f)$, $e\in C$,
$f\notin C$, contradicting the fact that $C(e)$ contained the minimum
number of links.

Hence we may assume that $C(e)$ contains no internal links; assume it
has some external links incident to $V_i$. Let $C(e)$ be of the form
$P_1-f_1-P_2-f_2-\ldots-P_t-f_t$, where $f_1,\ldots,f_t$ are the external
links incident to $V_i$, and $P_1,\ldots,P_j$ are the paths on $C(e)$
between two subsequent $f_j$'s.
If $t=2$, then the intersection between $C(e)$ and $(V_i,E_i)$ must
clearly be a path and hence the claim follows. Assume now $t>2$, and
that $e\in P_1$.
Let $Q$ denote
the path in $E_i$ between the endpoints of $f_1$ and $f_t$. Now
$f_1-Q-f_t-P_1$ gives a circuit in $E\cup F$ containing $e$, a
contradiction to the choice of $C(e)$.
\end{proof}

To define the notion of shadows in this setting, we first need
the analogues of $P(f)$ for external links.
This motivates our next definition.
Consider a leaf $u$ in a tree $(V_i,E_i)$ and let $(V_j,E_j)$ be a
different component. For some $1\le t\le \pe$, let $S_t(u,V_j)$ denote the endpoint of a
cheapest link between $u$ and a node in $V_j$ of weight at most $t$, that is
\[
S_t(u,V_j)=\mbox{argmin}_z\{c(f): f \mbox{ is an $(u,z)$ link}, z\in V_j,
w(f)\le t\}.
\]
If there are multiple possible choices, pick one arbitrarly.
 We say that the external
link $f=(u,v)$ is {\em foliate} if one of its endpoints, say $u$,  is a leaf in one of
the components $(V_i,E_i)$; let 
$w(f)=t$; assume $v\in V_j$ ($i\neq j$). Let
$P(f)$ denote the unique path in $E_j$ between $v$ and
$S_t(u,V_j)$.
As we shall see in Section~\ref{sec:completion}, foliate links with larger $P(f)$ are more useful, which motivates defining shadows based on comparing these sets. Shadows will be defined for internal links and foliate external
links only. All other external links are only shadows of themselves.

\begin{figure}
{\centering
\includegraphics[width=0.5\linewidth]{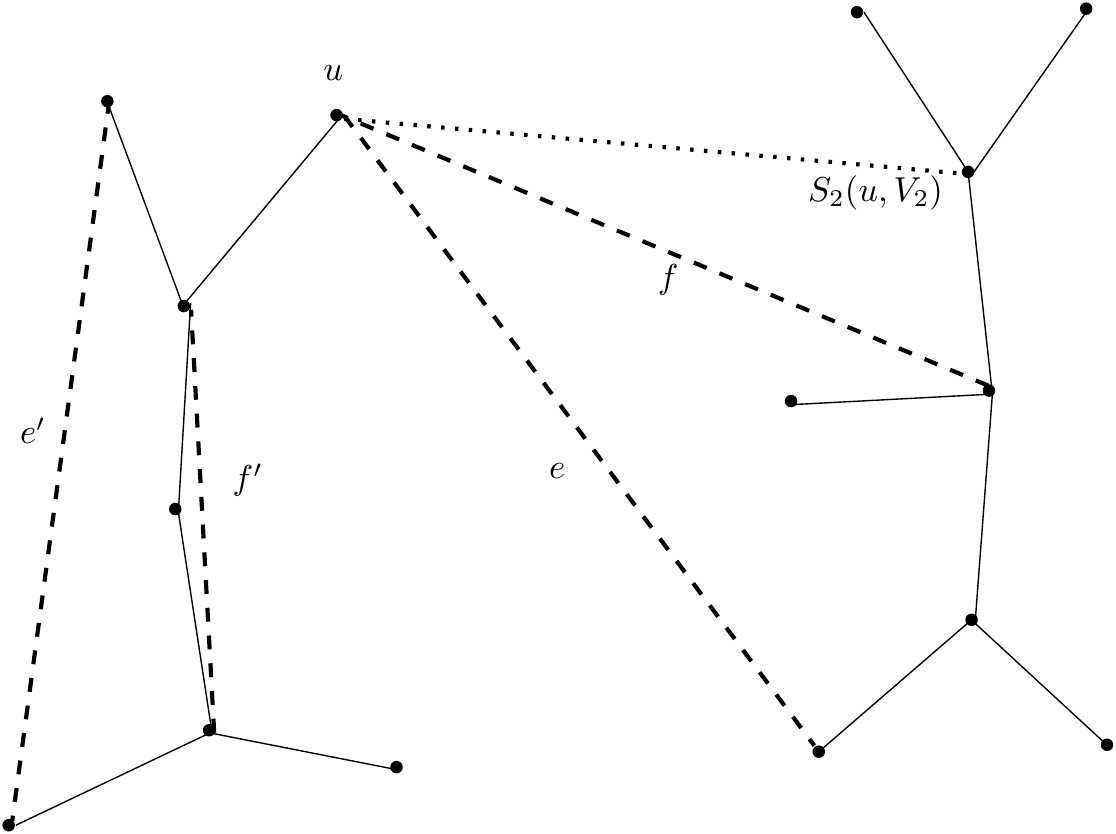}
\caption{The external link $f$ is a shadow of the external link $e$, and
  the internal link $f'$ is a shadow of the internal link $e'$.}\label{fig:shadows}
}
\end{figure}

\begin{definition}
Consider two links $e$ and $f$, with $w(f)\ge w(e)$.
 We say that $f$ is a {\em shadow} of $e$ in either of the following cases.
\begin{itemize}
\item $e=f$;
\item $e$ and $f$ are both internal links in the same component and
$P(f)\subseteq P(e)$;
\item $e=(u,x)$, $f=(u,y)$ are two foliate external links for a leaf
  $u$, and $P(f)\subseteq P(e)$.
\end{itemize}
\end{definition}
Note that in the last case, $x$ and $y$ must be in the same component $V_j$ not containing
  $u$, and for $t=w(e)$, $z=S_t(u,V_j)$, and $y$ is between $x$ and
  $z$. The definition is illustrated in Figure~\ref{fig:shadows}.
Given this notion, the definition
of metric instances is identical as in Section~\ref{sec:metric}. We say that the instance is {\em metric}, if 
\begin{enumerate}[(i)]
\item $c(f)\le c(e)$  holds whenever the link $f$ is a shadow of link $e$.
\item Consider three links $e=(u,v)$, $f=(v,z)$ and $h=(u,z)$ with
  $w(h)\ge w(e)+w(f)$. Then $c(h)\le c(e)+c(f)$.
\end{enumerate}

\subsection{Computing the metric completion}\label{sec:completion}
We use the algorithm {\em Metric-Completion}$(c)$ identical to the one
in Figure~\ref{fig:completion-alg}, with the meaning of shadows
modified. A technical difficulty is that the definition of shadow for
external links involve the nodes $S_t(u,V_j)$, whose definition
depends on the cost function, hence can change during the computation
of the metric completion. Moreover, the definition of $S_t(u,V_j)$
might involve an arbitrary choice if there are multiple cheapest
$t$-links. We use the following convention: while modifying the cost
function $c$, we modify the nodes $z=S_t(u,V_j)$ only if necessary.
That is, only if after the modification, link $(u,z)$ is not among the cheapest
$t$-links between $u$ and $V_j$ anymore.  We next prove that
Lemma~\ref{lem:metric-new} is still valid.

\begin{lemma}\label{lem:metric-new-2}
Consider a problem instance $(V,E,E^*,c,w,2,\pe)$.
The algorithm {\em Metric-Completion$(c)$} returns a metric cost function $\bar c$ with $\bar{c}(e)\le c(e)$ for every link $e\in E^*$. 
Moreover, if for a link set $\bar F\subseteq E^*$, $(V,E\cup \bar F)$ is
2-edge-connected, then there exists an  $F\subseteq E^*$
such that $(V,E\cup F)$ is 2-edge-connected, $c(F)\le \bar{c}(\bar F)$
and $w(F)\le w(\bar F)$. Consequently, and optimal solution for $\bar{c}$
provides an optimal solution for $c$.
\end{lemma}
\begin{proof}
  The proof of the metric property of $\bar c$ is almost identical to
  that in Lemma~\ref{lem:metric-new}.  We need only one additional
  observation: after fixing the triangle inequalities in iteration
  $t$, the nodes $S_t(u,V_j)$ cannot change anymore.  This is because
  all shadows of links between $u$ and $V_j$ are also links between
  $u$ and $V_j$, hence we cannot decrease the cost of
  the cheapest such link in the second part of phase $t$. Therefore, it follows that the shadow relations
  are unchanged during and after the second part of iteration $t$ and
  this relation is transitive.

  For the second part, it is again enough to verify the claim for the
  case when $\bar c$ arises by a single modification from $c$.  First,
  assume the modification is fixing a triangle inequality
  $c(h)>c(e)+c(f)$ by setting $\bar c(h)=c(e)+c(f)$ and $\bar
  c(g)=c(g)$ for every $g\neq h$. We again set $F=\bar F$ if $h\notin
  \bar F$ and $F=(\bar F\setminus\{h\})\dunion \{e,f\}$ otherwise. The
  only difference is that $F$ is a multiset (as in this section we
  assume that a link can be selected into the solution twice) and
  $\dunion$ denotes disjoint union, i.e. if $e$ or $f$ was already
  present in $F$, then we keep the old copies as well; but the same
  analysis carries over.

Next, assume $\bar c(f)=c(e)$ was set in the second part of iteration
$t$, and $\bar c(g)=c(g)$ for every $g\neq f$.
 If $f$ is an internal link, it is easy to verify that replacing $f$ by $e$ retains 2-edge-connectivity.

Let us now focus on the case when $f$ is a foliate external link, and
$f\in\bar F$.  Let
$e=(u,x)$, $f=(u,y)$, $t=w(f)$, with $u$ being a leaf, and 
$x,y\in V_j$ for a component not containing $u$; let $z=S_t(u,V_j)$.  Let $h=(u,z)$
denote a cheapest $t$-link between $u$ and $V_j$. As $f$ is a shadow
of $e$, the node $y$ appears on the path between $x$ and $z$ in $E_j$.

Let $F_{e}=(\bar F\dunion \{e\})\setminus\{f\}$ and $
F_{h}=(\bar F\dunion \{h\})\setminus\{f\}$. We aim to prove that either
$E\cup  F_e$ or $E\cup  F_h$ is 2-edge-connected. 
Let us say that an edge in $(E\cup \bar F)\setminus\{f\}$ is  {\em
$e$-critical} or  {\em
$h$-critical}, if it is a cut edge in $E\cup  F_e$ or in $E\cup
 F_h$, respectively. We call an edge {\em critical} if it is
either of the two.
\begin{myclaim}\label{cl:crit-on-path}
 If $g$ is $e$-critical, then it must lie on the path in $E_j$ between 
  $x$ and $y$. If $g$ is $h$-critical, then it must lie on the path in $E_j$
  between $y$ and $z$.
\end{myclaim}
\begin{proof}
  We prove for the $e$-critical case; the same argument works when $e$
  is  $h$-critical. 
Consider the circuit $C(g)$ containing a minimum number of links as
in Proposition~\ref{prop:intersect}.
For $g$ to become a cut edge in $E\cup F_e$, we must have
  $f\in C(g)$.
  Let $C'$ denote the circuit consisting of the links $e=(u,x)$, $f=(u,y)$ and the $x-y$ path on $E_j$. If the latter does not contain $g$, then we may use 
  Proposition~\ref{prop:strong-circuit} for $C(g)$ and $C'$ to obtain a circuit $C''\subseteq C(g)\cup C'$ with $g\in C''$, $f\notin C''$. The existence of such a $C''$
  contradicts our assumption that $g$ is a cut edge in $E\cup F_e$.
\cqed
\end{proof}

\begin{myclaim}
Either there exist no $e$-critical edges or there exist no $h$-critical edges.
\end{myclaim}
\begin{proof}
 For a contradiction, assume that there exists an $e$-critical edge $g_e$
  and an $h$-critical $g_h$.  Consider the circuits $C(g_e)$ and
  $C(g_h)$ containing the minimum number of links as in  Proposition~\ref{prop:intersect}; for the critical property, both of them must contain $f$. By Claim~\ref{cl:crit-on-path}, both $g_e,g_h\in E_j$; $g_e$ lies on the $x-y$ path, and $g_h$ lies on the $y-z$ path.
 Then Proposition~\ref{prop:intersect} implies that circuit $C(g_e)$ must
  be disjoint from the $y-z$ path in $E_j$ and $C(g_h)$ must be
  disjoint from the $x-y$ path.  Hence $g_h\notin C(g_e)$ and
  $g_e\notin C(g_h)$. Using Proposition~\ref{prop:strong-circuit}, we
  get a circuit $C\subseteq (C(g_e)\cup C(g_h))\setminus\{f\}$ and
  $g_e\in C$. This circuit is contained in $E\cup F_e$, a
  contradiction to the fact that $g_e$ is $e$-critical.  
\cqed
\end{proof} 
This claim
completes the proof, showing that $f$ can be exchanged to either $e$
or $h$. (It is easy to check that $e$ or $h$ itself cannot become a
cut edge, as it would imply that $f$ was a cut edge in $E\cup F$).
\end{proof}

\subsection{FPT algorithm for metric instances}\label{sec:metric-fpt}
In this section, we assume thaa problem instance $(V,E,E^*,c,w,2,\pe)$ is
metric. Let $R$ again denote the set of {\em corner nodes}, that is,
nodes of degree not equal to 2. Again, if there are more than $2\pe$
leaves, then the problem is infeasible; otherwise, $|R|\le 4\pe-2$.
For a leaf $u$ in the tree $(V_1,E_1)$.
\[
{\cal S}_u=\{ v\in V: v=S_t(u,V_j)\mbox{ for some }1\le t\le \pe, 2\le j\le r\}.
\]

The following theorem gives rise to a straightforward FPT algorithm.
\begin{theorem}\label{thm:degree-3-gen}
Consider a metric instance $(V,E,E^*,c,w,2,\pe)$, and let  $u$ be a leaf in 
the tree $(V_1,E_1)$. 
There exists an optimal solution 
solution $F$ such that for every link $f=(u,v)\in F$, it holds that $v\in R\cup {\cal S}_u$.
\end{theorem}
Given this theorem, the FPT algorithm is as follows. If the number of leaves is more than $2\pe$, we terminate by concluding infeasibility.
Otherwise, we pick an arbitrary node $u$ in the first tree. We branch according to all possible incident links connecting it to one of the corner
nodes or to the elements of ${\cal S}_u$. This is altogether $O(\pe)$ nodes with $\pe$ possible links connecting them to $u$, giving
$O(\pe^2)$ branches. This gives an algorithm with running time
$(\pe^2)^\pe=2^{O(\pe\log \pe)}$, proving Theorem~\ref{th:0to2}.

\begin{figure}
{\centering
\includegraphics[width=0.6\linewidth]{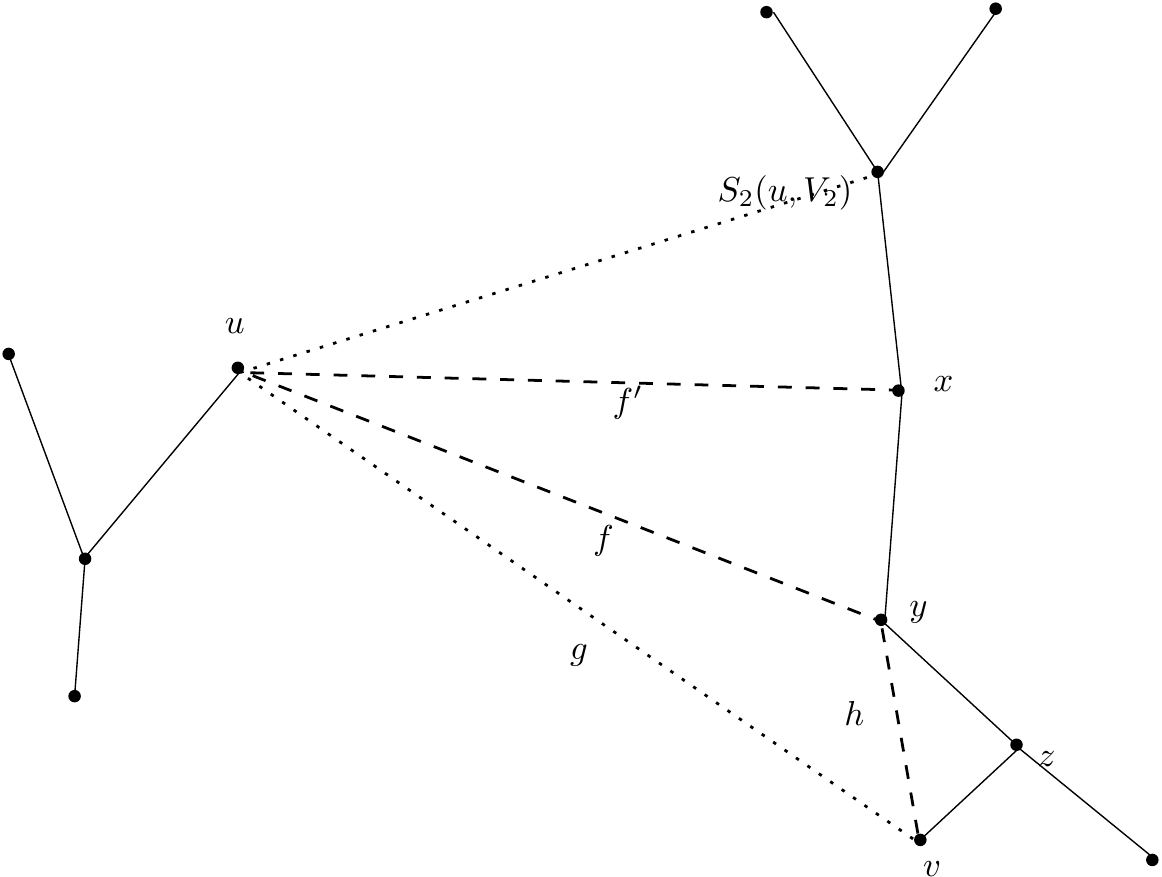}
\caption{Illustration of the proof of Theorem~\ref{thm:degree-3-gen}.}\label{fig:proof2}
}
\end{figure}

\begin{proof}[Proof of Theorem~\ref{thm:degree-3-gen}]
For an internal link or a foliate external link $f$, let $\ell(f)=|P(f)|$.
For other external links, let $\ell(f)=0$.
Consider an optimal solution $F$ such that $|F|$ is minimal, and
subject to this, $\ell(F)=\sum_{e\in F}\ell(f)$ is minimal.

For a contradiction, consider a link $f=(u,y)$ with $y\notin R\cup
{\cal S}_u$. Let $t=w(f)$.  If $f$ is an internal link, let $x$ be the
neighbour of $y$ between $u$ and $y$. If $f$ is external,
w.l.o.g. assume $y\in V_2$; in this case, let $x$ be the neighbour of
$y$ closer to $S_t(u,V_2)$.  In both cases, let $z$ be the other
neighbour of $y$ in $E_1$ or in $E_2$, which is uniquely defined since $y$ has degree 2.
Note that $\ell(f)$ is the length of the path between $u$ and $y$ in
$E_1$ or between $S_t(u,V_2)$ and $y$ in $E_2$. The external  case is
illustrated in Figure~\ref{fig:proof2}; for the internal case, see Figure~\ref{fig:proof1} in Section~\ref{sec:1-to-2}.

\begin{myclaim}\label{cl:must-f}
For any circuit $C\subseteq E\cup F$ with $xy\in C$, we must have $f\in C$.
\end{myclaim}
\begin{proof}
For a contradiction, assume there exists a circuit $C$ with $xy\in C$, $f\notin C$.
Let $f'=(u,x)$ be a $t$-link, and consider $F'=(F\setminus \{f\})\cup
\{f'\}$.  Link $f'$ is a shadow of $f$ and hence $c(f')\le c(f)$, that is,
$c(F)\le c(F')$; further, $\ell(f')=\ell(f)-1$. We claim that $E\cup
F'$ is also 2-edge-connected, thereby contradicting the minimal choice
of $\ell(F)$. The edge $xy$ is not a cut edge, as witnessed by the
circuit $C$ not containing $f$. For any other edge $e\in (E \cup F)\setminus\{f\}$, we know that there is a circuit $C(e)\subseteq E\cup F$
containg $e$. If $f\notin C(e)$, then $C(e)\subseteq E\cup F'$ as well.
In the sequel, assume $f\in C(e)$. If  $xy\in C(e)$, then 
$f$ and $xy$ can be replaced in $C(e)$ by $f'$, giving a circuit in $E\cup F'$ containing $e$.
On the other hand, if $xy\notin C(e)$, then we can replace $f$ by $f'$
and $xy$.
We can show in a similar way that $f'$ cannot be a cut edge either: given a circuit of
$E\cup F$ containing $f$, we can either replace $f$ by $f'$ and $xy$, or replace $f$ and $xy$ by $f'$ to obtain a circuit in
$E\cup F'$ containing $f'$.
\cqed
\end{proof}

Consider now the edge $yz\in E$,  and let $C(yz)$ be a circuit in
$E\cup F$ containing $yz$ and having a minimal number of 
links. Let $C(xy)$ be the analogous circuit for $xy$; the previous
claim implies $f\in C(xy)$.
\begin{myclaim}
We have $xy,f\notin C(yz)$, and there is a link $h=(y,v)\in C(yz)\cap F$.
\end{myclaim}
\begin{proof}
By Proposition~\ref{prop:intersect}, $C(yz)$ intersects the
 component of $yz$  ($E_1$ or $E_2$) in a single path $P$ and there
 are at most two incident links. If $xy\in P$, then by the previous claim, $f\in C(yz)$. Then $y$ has degree 3 in the circuit $C(yz)$, a contradiction.
 Consequently, the path $P$ must end in $y$, and hence $C(yz)$ must
 contain a link $h$ incident to $y$. The proof is complete by showing $h\neq
 f$. Indeed, if $h=f$, then we can apply
 Proposition~\ref{prop:strong-circuit} for $C(xy)$ and $C(yz)$ to
 obtain a circuit $C'$ with $xy\in C'$ and $f\notin C'$, a
 contradiction to the previous claim.
\cqed
\end{proof}

The rest of the proof is dedicated to showing that
2-edge-connectivity is maintained if we replace $f$ and $h$ by a
$(u,v)$-link $g$ of weight $w(f)+w(h)$. Since the instance is metric,
we must have $c(g)\le c(f)+c(h)$. Let $F'=(F\cup\{g\})\setminus\{f,h\}$.  Showing that $E\cup F'$ is
2-edge-connected yields a contradiction to the minimal choice of
$|F|$.  By Proposition~\ref{prop:2-ec-gen}, we have to show that
$E\cup F'$ is connected and  for each edge there is a circuit
containing it.  Connectivity follows easily: if $E\cup F'$ became
disconnected by removing  links $f=(u,y)$ and $h=(y,v)$, and adding $(u,v)$, then node $y$ must lie in a different component than $u$ and
$v$. However, as $f\in C(xy)$ by Claim~\ref{cl:must-f}, the path
$C(xy)\setminus\{f\}$ still appears in $E\cup F'$ and connects the
endpoints $u$ and $y$ of $f$.  To verify the existence of a circuit
for each edge, we need the following.
\begin{myclaim}\label{claim:only-v}
The only common node of the  circuits $C(xy)$ and $C(yz)$ is $y$.
\end{myclaim}
\begin{proof}
For a contradiction, assume the two circuits intersect in nodes other
than $y$.
Let us start moving on the path $P_0=C(xy)\setminus\{f\}$ from $y$ until
we hit the first node on 
$C(yz)$; let $a$ be this intersection point and let $P_1$ be the part
of $P_0$ between $y$ and $a$.
 Let $P_2$ be one of the two parts of $C(yz)$ between $a$ and $y$.
Now $P_1\cup P_2$ is a circuit containing $xy$ but not $f$, a
contradiction to   Claim~\ref{cl:must-f}.
\cqed
\end{proof}

Consequently, $\hat C= (C(xy)\cup C(yz)\cup\{g\})\setminus\{f,h\}$ is
a circuit in $E\cup F'$ containing $g$.  For an arbitrary $e\in (E\cup F')\setminus\{g\}$,
consider the circuit $C(e)$ in $E\cup F$. We are done if $C(e)$
contains neither of $f$ and $h$. If $C(e)$ contains both $f$ and $h$,
then we can replace these two edges in the circuit with $g$. Assume
$C(e)$ contains exactly one of them, say $f\in C(e)$ (the case $h\in C(e)$ can
be proved similarly). If $e\in C(xy)$, then $\hat C$ does contain $e$.
Otherwise, if $e\notin C(xy)$, then we may use
Proposition~\ref{prop:strong-circuit} to obtain a circuit $C'\subseteq
C(e)\cup C(xy)$, $f\notin C'$, $e\in C'$. Also, $h\notin C'$ as it
was contained in neither $C(e)$ nor $C(xy)$. Now $e\in C'\subseteq
E\cup F'$, completing the proof.  
\end{proof}

\subsection{Forbidding using links twice}\label{sec:no-parallel}
The algorithm presented in the previous section solves the version of
the problem where we allow taking the same link twice in a
solution. Here we show how to solve the original version of the
problem, where this is not allowed.  We present a solution for the
restriction when we do not even allow adding parallel links of
different weights between two nodes. The argument can be easily
modified to the weaker restriction when we may allow parallel links
with different weight.

As before, let $(V_1,E_1),\dots,(V_r,E_r)$ be the components of the
input graph. Note that the components can contain cycles and hence
they are not necessarily trees; however, this will not cause any
complications for the arguments presented in this section. For every
$1\le i < j \le r$ and $1\le t \le p$, consider the $t$-links between
$V_i$ and $V_j$; if there is a {\em unique} $t$-link of minimum cost
between these components, then let us select this link into the set
$S$.  If $r>\pe$, then there exists no feasible solution (as we would
need more than $\pe$ links to connect the components).  If $r\le \pe$,
we have $|S|\le \pe^3$.

As a first step of the algorithm, we branch on which subset of $S$
appears in the solution. That is, for every subset $S'\subseteq S$
with $w(S')\le \pe$ and not containing any parallel links, we obtain a
new graph $G'$ by adding the links in $S'$ to the graph $G$.  Note that
adding the set $S'$ can decrease the number of components and can
create further cycles. We define a new parameter $\pe'=\pe-w(S')$ and
define a new cost function $c'$, whose only difference from $c$ is
that the cost of every link in $S\setminus S'$ and of every link
parallel to a link in $S'$ is $\infty$. We solve the modified instance
for the graph $G'=(V',E')=(V,E\cup S')$, parameter $\pe'$, and cost
function $c'$ using the algorithm of the previous section. If $F'$ is
the solution obtained this way, then we return the solution $F=S'\cup
F'$. The branching step adds a factor of $O((\pe^3)^\pe)=O(2^{\pe\log
  \pe})$ to the running time of the algorithm.

 It is clear that if the original instance has a solution not
containing duplicated links, then no matter which subset of the links
$S$ it uses, our algorithm returns a solution with not larger
cost. More importantly, we claim that if our algorithm returns a
solution using some links twice, then it can be modified such that it
does not use any link twice and the cost does not increase. These two
statements prove that this algorithm indeed finds an optimum solution
for the problem where duplicated links are not allowed. 
We observe first the following simple lemma:
\begin{lemma}\label{lem:replacetwice}
  Let $G=(V,E \cup F)$ be a 2-edge-connected graph. {\em(i)} If $F$
  contains two parallel edges with the same endpoints $x$ and $y$ in the same
  component of $G$, we may remove one of them without destroying
  2-edge-connectivity.  {\em(ii)} Suppose that $e_1,e_2 \in F$ are two
  parallel links with endpoints $x$ and $y$ in different components of
  $G$. Suppose that $F$ contains another link $e^*$ (different from
  $e_1,e_2$) whose endpoints are in the same connected components of
  $(V,E)$ as $x$ and $y$, respectively. Then $G^*=(V,E\cup (F\setminus
  e_1))$ is also 2-edge-connected.
\end{lemma}
\begin{proof}
  The first statement is straightforward. For the second, observe that
  the edge $e_2$ is not a cut edge in $G^*$: there is a circuit
  containing $xy$ formed by $e^*$, a path in the component of $x$, a
  path in the component of $y$, and $e_2$ itself. Moreover, if $G^*$
  has a cut edge other than $e_2$, then it is a cut edge of $G$ as
  well, a contradiction.  
\end{proof} 
Note that $F'$ cannot contain links parallel to $S'$ (as the cost of
every such link is $\infty$ in $c'$), hence parallel links can appear
only in $F'$ itself.  Suppose that the algorithm finds a multiset $F'$
of links, containing parallel pairs.  Consider two links $e_1,e_2\in
F'$ between $x$ and $y$; let $t=w(e_1)$.  If $x$ and $y$ are in the
same component of $V_i$ of $G$, then Lemma~\ref{lem:replacetwice}(i)
implies that $e_1$ can be safely removed.  Assume therefore that $x$
and $y$ are in two different connected components $V_i$ and $V_j$ of
$G$, respectively.  As $e_1\notin S$, link $e_1$ is not the unique
minimum cost $t$-link between $V_i$ and $V_j$ in the original
instance. Therefore, there is a $t$-link $e^*$ between $V_i$ and $V_j$
with $c(e^*)\le c(e_1)$ (note that possibly $e^*\in S\setminus
S'$). Link $e^*$ connects the same two connected components of $G$ as
$e_1$. Therefore, if $e^*$ is already in $S'\cup F'$ or there is a
link parallel to it in $S'\cup F'$, then
Lemma~\ref{lem:replacetwice}(ii) implies that removing $e_1$ from
$S'\cup F'$ does not destroy 2-edge-connectivity.  Otherwise, we
replace $e_1$ with $e^*$; the cost of the new solution $S'\cup
(F'\setminus e_1)\cup e^*$ obtained this way is not larger than the
cost of $S'\cup F'$. Again by Lemma~\ref{lem:replacetwice}(ii),
removing $e_1$ from $(V,E\cup S'\cup (F'\cup e^*))$ does not destroy
2-edge-connectivity, i.e., $(V,E\cup S'\cup (F'\setminus e_1)\cup e^*)$ is
2-edge-connected. We repeat this replacement for every duplicated
link. Note that this process does not create new duplicated links: we
add $e^*$ only if there is no link in $F=S'\cup F'$ with the same
endpoints as $e^*$. Therefore, we obtain a solution having not larger
cost and containing no duplicated links.

\fi
